\documentclass[10pt,twoside]{article}
\usepackage{amssymb}
\usepackage{latexsym}
\usepackage{amsfonts}
\usepackage{graphicx}
\usepackage{amsmath,amsthm,amssymb}
\allowdisplaybreaks

\newtheorem{theorem}{Theorem}[section]

\newtheorem{lemma}{Lemma}[section]
\newtheorem{proposition}{Proposition}[section]
\newtheorem{remark}{Remark}[section]

\newcommand{\ds}{\displaystyle}

\def\theequation{\thesection.\arabic{equation}}
\def \[{\begin{equation}}
\def \]{\end{equation}}

\newif \ifLastSection \LastSectionfalse

\renewcommand{\theequation}{\ifLastSection A\arabic{equation}\else
\thesection .\arabic{equation}\fi}

\newcommand{\R}{{\mathbb R}}

\makeatletter\def\theequation{\arabic{section}.\arabic{equation}}\makeatother
\begin{document}
\title{
{\bf\Large NODAL Vector  solutions with clustered peaks for a nonlinear elliptic equations in $\R^3$}\footnote{The authors sincerely thank Professor S.J. Peng for helpful discussions and suggestions. This work was partially supported by NSFC
(No.11301204; No.11371159; No.11101171), the phD specialized grant of the Ministry of Education of China(20110144110001).}\hspace{2mm}\\}
\author{{\bf\large Qihan  He}\vspace{1mm}\\{\it\small School of Mathematics and Statistics, }\\ {\it\small Central China Normal
University},
{\it\small  Wuhan, 430079, P. R. China}\\
{\it\small e-mail: heqihan277@163.com}\vspace{1mm}\\
{\bf\large Chunhua Wang}\footnote{Corresponding author}\hspace{2mm}\\
{\it\small School of Mathematics and Statistics, }\\ {\it\small Central China Normal
University},
{\it\small  Wuhan, 430079, P. R. China}\\
{\it\small e-mail: chunhuawang@mail.ccnu.edu.cn}}\vspace{1mm}
\maketitle
\begin{center}
{\bf\small Abstract}
\vspace{3mm}
\hspace{.05in}\parbox{4.5in}
{{\small In this paper, we study the following coupled nonlinear Schr\"{o}dinger system in $\R^3$
$$
\left\{%
\begin{array}{ll}
-\epsilon^2\Delta u +P(x)u=\mu_1 u^3+\beta v^2u,~~&x\in \R^3,\vspace{0.15cm}\\
-\epsilon^2\Delta v +Q(x)v=\mu_2 v^3+\beta u^2v,~~&x\in \R^3,\\
\end{array}%
\right.
$$
where  $\mu_1 >0,\mu_2>0$ and $\beta \in \R$ is a coupling constant. Whether the system is repulsive or
attractive, we prove that it has nodal semi-classical segregated or synchronized bound states with clustered spikes
for sufficiently small $\epsilon$ under some additional conditions on $P(x), Q(x)$ and $\beta$. Moreover, the number of
this type of solutions will go to infinity as $\epsilon \to 0^+$.}}
\end{center}
\section{\bf Introduction}\label{s1}
In this paper, we consider the following nonlinear Schr\"{o}dinger system in $\R^3$
\begin{equation}\label{1.1}\left\{%
\begin{array}{ll}
-\epsilon^2\Delta u +P(x)u=\mu_1 u^3+\beta v^2u,~~&x\in \R^3,\vspace{0.15cm}\\
-\epsilon^2\Delta v +Q(x)v=\mu_2 v^3+\beta u^2v,~~&x\in \R^3,\\
\end{array}%
\right.\end{equation}
where we assume that $P(x)$ and $Q(x)$ are   continuous bounded  radial
functions, $\mu_1 >0,\mu_2>0$ and $\beta \in \R$ is a coupling constant.

It motivates us to study problem \eqref{1.1} that we look for standing-wave solutions for the following time-dependent
coupled nonlinear Schr\"{o}dinger system:
\begin{equation}\label{P12}\left\{%
\begin{array}{ll}
i\epsilon\frac{\partial \psi}{\partial t}= -\frac{\epsilon^2}{2m}\Delta_x \psi+P(x) \psi-\mu_1 |\psi|^2\psi-\beta |\phi|^2\psi,~~&x\in \R^3,t>0,\vspace{0.15cm}\\
i\epsilon\frac{\partial \phi}{\partial t}= -\frac{\epsilon^2}{2m}\Delta_x \phi +Q(x)\phi-\mu_2 |\phi|^2\phi-\beta |\psi|^2\phi,~~&x\in \R^3,t>0,\vspace{0.15cm}\\
\psi=\psi(x,t) \in \mathcal{C}, \phi=\phi(x,t) \in \mathcal{C},
\end{array}%
\right.\end{equation}
which models a binary mixture of Bose-Einstein condensates in two different hiperfine states (see \cite{EGBB,MBGCW,HMEWC,T}), and where $\epsilon$ is the plank constant, $m$ is the atom mass, $P(x)$ and $Q(x)$ are the trapping potentials for two hyperfine states, respectively;
the constants $\mu_1$ and $\mu_2$ represent the intraspecies scattering lengths and $\beta$ is the interspecies scattering length. The sign of the interspecies
scattering length determines whether the interaction of states are repulsive or attractive. If $\beta>0$, the interaction is attractive, and the components of
a vector solutions lead to synchronize. On the other hand, if $\beta<0$, the interaction is repulsive, leading to phase separations. These phenomena have been confirmed in experiments and in numeric simulations (see \cite{MBGCW,CLLL,HMEWC,GHZ} and references therein).
Problem \eqref{P12}, also known as Gross-Pitaevskii equations, arises in many applications. For example, in some problems arising in nonlinear optics, in plasma physics and in the condensed matter physics. Physically, $\psi$ and $\phi$ are the corresponding condensated wave functions (see \cite{AA}).

This system \eqref{1.1} has been extensively investigated under various assumptions on $P(x),Q(x)$ and $\beta$ in recent years (see \cite{A}
\cite{AC}-\cite{BDW},\cite{BW}-\cite{EGBB},\cite{CNY}-\cite{CTV},\cite{DW}-\cite{P},\cite{S,TV,W,W1} and therein ).
Here we want to mention some significant works.
In \cite{LW2}, no matter the interspecies scattering length $\beta$ is positive or negative, Lin and Wei have obtained
least energy solutions for the two coupled nonlinear Schr\"{o}dinger system with the trap potentials by using Nehari's manifold  and derived their asymptotic behaviors by some techniques of
singular perturbation problem. At the same time,  Chen, Lin and Wei \cite{CLW} have proved the existence of the positive solutions with any prescribed spikes by
the reduction methods. In \cite{A}, Alves has been concerned with the existence and the concentration of positive solutions by the mountain pass theorem.
  Wan \cite{W} used the category theory to study the multiplicity of positive solutions and their limiting behavior as $\epsilon \to 0^+$. Also in \cite{W1},  Wan and \'{A}vila utilized the category theory  studying the relation between the number of positive standing waves solutions for
the special system \eqref{1.1} with $P(x)=Q(x)$ and $\beta=0$ in $\R^N$ and the topology of the set of minimum points of potentials.  Pomponio
in \cite{P} also has proved the existence of concentrating solutions for a general system with repulsive interaction of states
and that how the location of the concentration points depends strictly on the potentials. In \cite{BDW}, Bartsch, Dancer and Wang considered the repulsive case and obtained segregated radial solutions by global bifurcation methods for the the general systems \eqref{1.1}, establishing the existence of infinite branches of radial solutions with the property that $\sqrt{\mu_1-\beta}\psi-\sqrt{\mu_2-\beta}\phi$ has exactly $k$ nodal domains for solutions along the
kth branch. Recently, Pi and Wang \cite{PWc} have constructed multiple solutions
with any prescribed spikes and proved that the spikes will approach the local maximum point of the trap potentials as $\epsilon \to 0^+$.

Here we should point out that in  the results mentioned above, the solutions are positive. Although there is a wide literature studying existence, multiplicity and shape
of positive solutions, there are few papers dealing with the case of nodal solutions, with the exception of the single Schr\"{o}dinger equations for the one-dimensional case or the radial case(\cite{BW1}) which allows methods, like the use of a natural constraint, which do not work in the nonradial setting considered here.

As far as we know, there are no results on the existence of nodal non-radial semi-classical bound  states which have any prescribed nodal domain. In this paper, we
will present some results which contributes to this respect.

In order to state our main results, first we
 assume that $\inf\limits_{r \geq 0}P(r)>0$ and $\inf\limits_{r \geq 0}Q(r)>0$ satisfy the following conditions:

(~P~):~~There are constants $a \in \R, m>1~\hbox{and}~\theta>0$, such that as $r\rightarrow 0^+$
$$P(r)=1+ar^m+O(r^{m+\theta}).$$

(~Q~):~~There are constants $b \in \R, n>1~\hbox{and}~\delta>0$, such that as $r\rightarrow 0^+$

$$Q(r)=1+br^n+O(r^{n+\delta}).$$

The main results of our paper are as follows.
\begin{theorem}\label{Th1}~Let (P) and (Q) hold. Then for any fixed $ k \in N^+$, there exists
a decreasing sequence $\{\beta_l\}\subset (-\sqrt{\mu_1 \mu_2},0)$ with $\beta_l \rightarrow -\sqrt{\mu_1 \mu_2}$ as $l\rightarrow \infty$ and
$\epsilon_0>0$
such that  for $\beta \in (-\sqrt{\mu_1 \mu_2},0)\cup (0,\min\{\mu_1,\mu_2\})\cup(\max\{\mu_1,\mu_2\},\infty)$
and $\beta\neq \beta_l$, and $0<\epsilon<\epsilon_0,$ \eqref{1.1} has a vector solution $(u_\epsilon,~v_\epsilon)$ with $k$ positive peaks and $k$ negative peaks, and
the peaks of the solution  approaching to the extremal  point $0$ of $P(x)$ and $Q(x)$ provided one of the following two conditions holds:\\
(1)\quad  $m<n,a>0~ \hbox{and}~b \in \R; \hbox{or}~ m>n, a\in \R ~\hbox{and} ~b>0;$\\
(2)\quad $m=n, aB +bC_0 >0, \hbox{where}~ B~\hbox{and}~ C~\hbox{are~defined ~in ~Proposition \ref{proA1}};$\\
Furthermore,
$$
\|\sqrt{|\mu_1-\beta|}u_\epsilon-\sqrt{|\mu_2-\beta|}v_\epsilon\|_{H^1}+\|\sqrt{|\mu_1-\beta|}u_\epsilon-\sqrt{|\mu_2-\beta|}v_\epsilon\|_{L^\infty} \to 0,~~~as~~\epsilon \to 0^+.
$$
\end{theorem}

\begin{theorem}\label{Th2} Let (P) and(Q) hold. If $m=n, a>0, b>0,$ then for any fixed $ k\in N^+$, there exist
constants $\beta_0>0$ and $\epsilon_0>0$ such that for any $\beta<\beta_0$ and $0<\epsilon<\epsilon_0,$ \eqref{1.1} has a vector solution $(\widetilde{u}_\epsilon,~\widetilde{v}_\epsilon)$
with $k$ positive peaks and $k$ negative peaks which approach to the local minimum point $0$ of $P(x)$ and $Q(x)$ as $\epsilon \to 0^+$. Furthermore,
$$\|\sqrt{\mu_2}\widetilde{u}_\epsilon(.)-\sqrt{\mu_1}\widetilde{v}_\epsilon(T_\epsilon.)\|_{H^1}+\|\sqrt{\mu_2}\widetilde{u}_\epsilon(.)-\sqrt{\mu_1}\widetilde{v}_\epsilon(T_\epsilon.)\|_{L^\infty} \to 0,\quad \quad as \quad \epsilon \to 0^+.$$
Here $T_\epsilon \in SO(3)$ is the rotation on the $(x_1,x_2)$ plane of $\frac{\pi}{k}.$
\end{theorem}

Next, we introduce some notations to be used in the proofs of  the main results and formulate a version of the main results which give more precise
descriptions about the segregated and synchronized character of the solutions. In doing so ,we also outline the main idea and the approaches in the proofs of
Theorems \ref{Th1} and \ref{Th2}.

Define
\begin{equation}\label{1.2}\begin{split}
H_s&=\bigg\{u\in H^1(\R^3): u~\hbox{is~even~in~}y_h, h=2,3,\\
&~~~~~~~~u\Big(r\cos\big(\theta+\frac{\pi j}{k}\big), r\sin\big(\theta+\frac{\pi j}{k}\big),x_3\Big)
=(-1)^ju\left(r\cos\theta, r\sin\theta,x_3\right)\bigg\},
\end{split}
\end{equation}
where $H^1(\R^3)$ is the usual Sobolev space with the norm for any bounded function $K(x)$
$$\|u\|^2_{\epsilon,K}=(u,u)_\epsilon=\int_{\R^3}(\epsilon^2|\nabla u|^2 +K(x)|u|^2)dx,$$
and define  $ H  = H_s \times H_s$   endowed with the following norm

$$\|(u,v)\|^2_\epsilon=\|u\|^2_{\epsilon,P} +\|v\|^2_{\epsilon,Q}. $$

Set
$$
w_{y,\epsilon}(x)=w\big(\frac{x-y}{\epsilon}\big)
$$
and
\begin{equation}\label{s} S_\epsilon:=\Big[\frac{\min\{m, n\}-\delta}{2\sin\frac{\pi}{2k}}\epsilon\ln\frac{1}{\epsilon},\quad \frac{\min\{m, n\}+\delta}{2\sin\frac{\pi}{2k}}\epsilon\ln\frac{1}{\epsilon}\Big],
\end{equation}
where $\delta \in ( 0, \frac{\sigma}{1+\sigma} \min\{n,m\}),$ and $\sigma$ will be defined in  Proposition \ref{proA2}.
Denote
\begin{equation}\label{1.3}
 x^j:=\Big(r\cos\frac{(j-1)\pi}{k},~r\sin\frac{(j-1)\pi}{k},~x_3\Big),~j=1,2,\cdots,2k,~~r \in S_\epsilon.
\end{equation}
It is well-known that the following problem has a unique radial solution denoted by $w$
\begin{equation}\label{11.4}
-\Delta u+u =u^3, \max\limits_{x\in\R^3}u(x)=u(0),u>0,
\end{equation}
and the solution $w$ satisfies the following properties:
$$
w^\prime(r)<0,\,\,
    \lim\limits_{r\rightarrow\infty}r^{\frac{N-1}{2}}e^rw(r)=C_0>0,\,\,
 \lim\limits_{r\rightarrow\infty}\frac{w^\prime(r)}{w(r)}=-1.
$$

When $-\sqrt{\mu_1\mu_2}<\beta<\min\{\mu_1,\mu_2\}$ or $\beta>\max\{\mu_1,\mu_2\}$,
$(U,V):=(\alpha w, \gamma w)$ is a solution of
the following system:
\begin{equation}\label{sys}\left\{%
\begin{array}{ll}
-\Delta u +u =\mu_1u^3 +\beta v^2u,\,\,&x\in \R^{3},\vspace{0.15cm}\\
-\Delta v +v =\mu_2v^3 +\beta u^2v,\,\,&x\in \R^{3},\\
\end{array}
\right.
\end{equation}
where $\alpha=\sqrt{\frac{\mu_2-\beta}{\mu_1\mu_2-\beta^2}},\gamma=\sqrt{\frac{\mu_1-\beta}{\mu_1\mu_2-\beta^2}}.$

We let

$$U_r(x)=\sum\limits_{j=1}^{2k}(-1)^{j-1}U_{x^j,\epsilon},
V_r(x)=\sum\limits_{j=1}^{2k}(-1)^{j-1}V_{x^j,\epsilon}.
$$

We will verify Theorem \ref{1.1} by proving the following result:

\begin{theorem}\label{Th3} Under the assumptions of Theorem \ref{Th1}, there exists a positive constant $\epsilon_0>0$
such that for any $0<\epsilon<\epsilon_0$, \eqref{1.1} has a solution of the form
$$(u_\epsilon,~v_\epsilon)=(U_r(x)+\varphi(x),~V_r(x)+\psi(x)),$$
where $(\varphi(x),\psi(x))\in H$ and
$$\|(\varphi(x),\psi(x))\|_\epsilon=O\big(\epsilon^{\frac{3+\min\{m,n\}-\sigma}{2}}\big),~~|x^j|
=O\Big(\epsilon\ln\frac{1}{\epsilon}\Big)$$
for some small constant $\sigma>0$.
\end{theorem}

Let $U_i$ be the unique radial solution  of the following problem
$$-\Delta u+u =\mu_iu^3,\,\,\, \max\limits_{x\in\R^3}u(x)=u(0),u>0.$$
It is well known that $U_i$ is non-degenerate and $U_i^\prime(r)<0,\,\,
    \lim\limits_{r\rightarrow\infty}r^{\frac{N-1}{2}}e^rU_i(r)=C_0>0,\,\,
 \lim\limits_{r\rightarrow\infty}\frac{U_i^\prime(r)}{U_i(r)}=-1.$

 We will use $(U_1,U_2)$ to build up the approximate solutions for \eqref{1.1}.

Let $x^j$ be defined in \eqref{1.3} and denote
\begin{equation}\label{1.4}
y^j:=\Big(\rho\cos\frac{(2j-1)\pi}{2k},~\rho\sin\frac{(2j-1)\pi}{2k},~x_3\Big),j=1,2,\cdots,2k,
\end{equation}
where $\rho \in S_\epsilon.$

Let
\begin{equation}\label{1.5}
\tilde{U}_r=\sum\limits_{j=1}^{2k}(-1)^{j-1}U_{1,x^j,\epsilon},\,\,
\tilde{V}_\rho=\sum\limits_{j=1}^{2k}(-1)^{j-1}U_{2,y^j,\epsilon}.
\end{equation}

To prove Theorem \ref{Th2}, we need to prove the following result.

\begin{theorem}\label{Th4}  Under the assumptions of Theorem \ref{Th2},  there exists a positive constant
$\epsilon_0$ such that for any $0<\epsilon<\epsilon_0$, \eqref{1.1} has a solution of the form
$$(\widetilde{u}_\epsilon,~\widetilde{v}_\epsilon)=\big(\tilde{U}_r(x)+\tilde{\varphi}(x),~\tilde{V}_\rho(x)+\tilde{\psi}(x)\big),$$
where $(\tilde{\varphi}(x),\tilde{\psi}(x))\in H$ and
$$
\|(\tilde{\varphi}(x),\tilde{\psi}(x))\|_\epsilon=O\big(\epsilon^{\frac{3+\min\{m,n\}-\sigma}{2}}\big),~~|x^j|
=O\Big(\epsilon\ln\frac{1}{\epsilon}\Big),~~|y^j|=O\Big(\epsilon\ln\frac{1}{\epsilon}\Big)
$$
for some small constant $\sigma>0$.
\end{theorem}

\begin{remark}
Radial symmetries can be replaced by the following weaker symmetrical
assumptions: after suitably rotating the coordinate system,

$(P')~~P(x)=P(x',x_3)=P(|x'-\bar x'|,x_3-\bar
x_3)$ and $P(x)$ has the following expansion:
$$
P(r)=P(\bar x)+a|x-\bar x|^m+O(|x-\bar x|^{m+\theta}) ~\hbox{as}~
|x-\bar x|\to 0,
$$
where $\bar x\in\R^3$,  $a \in \R, m>1,\theta>0~\hbox{and}~P(\bar x)>0$
are  constants.

$(Q')~~Q(x)=Q(x',x_3)=Q(|x'-\bar x'|,x_3-\bar
x_3)$ and $Q(x)$ has the following expansion:
$$
Q(r)=Q(\bar x)+b|x-\bar x|^n+O(|x-\bar x|^{n+\delta}) ~\hbox{as}~
|x-\bar x|\to 0,
$$
where $\bar x\in\R^3$,  $b \in \R, n>1,\delta>0~\hbox{and}~Q(\bar x)>0$
are   constants.

\end{remark}

\begin{remark}
For $N=2$, if we adjust the constants $ \delta, \tau, \tau_2$ in \eqref{s}, then both Lemma \ref{lemma2.5}
and Proposition \ref{pro2.6} still hold.
In order to guarantee that Proposition \ref{pro2.6} holds, we can find nodal synchronize solutions of \eqref{1.1}
for the attractive case under the same assumptions.
 However, for the  
 repulsive  case, we can't find nodal segregated solutions of
 \eqref{1.1}, since Proposition \ref{pro3.4} can not hold.
\end{remark}

The proofs of our main result are based on the well-known
Lyapunov-Schmidt reduction procedure. In particular, in order to
deal with nodal clustered solutions, we perform the reduction in
suitable symmetric settings in the spirit of \cite{WW} where
infinitely many positive non-radial solutions for nonlinear
Schr\"{o}dinger equations were obtained. For the attractive case, we will construct nodal synchronize
solutions approximately as $\Big(\sum\limits_{j=1}^{2k}(-1)^{j-1}U_{x^j,\epsilon},~~~~
\sum\limits_{j=1}^{2k}(-1)^{j-1}V_{x^j,\epsilon}\Big)$ with the points $x^j$ locating on and dividing
equally  the circle with radius $C \epsilon\ln\frac{1}{\epsilon}$ into
$2k$ parts. Since the distance between two neighbor peaks with the
same sign is larger than that between two neighbor peaks with
opposite sign, the interaction among peaks with opposite sign
dominates that among peaks with the same sign. Hence, if  the slower decaying functions between  $Q(x)$ and $P(x)$ has local
minimum at the center of the circle, we can easily conclude that the
equilibrium is achieved for a suitable configuration of the points
$x^j$, which can be reduced to solve a minimization problem related
to energy functional. Generally speaking, the key to construct nodal
solutions by the reduction argument is to compare the influence
between the interaction  among the peaks with the same sign and that
among the peaks with opposite sign, the idea in \cite{WW} can help
us to construct a symmetric configuration space consisting of
$x^j\,(j=1,\cdots,2k)$ and hence realize the key.
For the repulsive case, we will construct nodal segregated
solutions approximately as $\big(\sum\limits_{j=1}^{2k}(-1)^{j-1}U_{1,x^j,\epsilon},~
\sum\limits_{j=1}^{2k}(-1)^{j-1}U_{2,y^j,\epsilon}\big)$ with the points $x^j$ and $y^j$ locating on and dividing
equally  the circles with radius $C_1 \epsilon\ln\frac{1}{\epsilon}$ and $C_2 \epsilon\ln\frac{1}{\epsilon}$ into
$2k$ parts, respectively and vector  $\overrightarrow{oy^j}$ dividing  equally angle $\angle x^jox^{j+1}$.  Then using the similar methods like the attractive
case, we can construct nodal segregated solutions.
 This idea is also
effective in finding infinitely many  non-radial positive solutions
for semilinear elliptic problems  (see,  \cite{PW}).

This paper is organized as follows. In section \ref{s2}, we will study the
finite-dimensional reduced problem and prove Theorem~\ref{Th3}. We will put
the study of the existence of segregated solutions for system \eqref{1.1} and the proof of the Theorem~\ref{Th4} into Section \ref{s3}.
Finally we will give all the technical calculations in the Appendix.

\section{\bf Synchronized Vector Solutions and the proof of Theorem \ref{Th1}}\label{s2}
\def\theequation{2.\arabic{equation}}\makeatother
\setcounter{equation}{0}

In this section we consider synchronized vector solutions and prove Theorem \ref{Th1} by proving Theorem \ref{Th3}.
The functional corresponding to \eqref{1.1} is
\begin{equation}\label{2.1}
\begin{array}{rl}
I_\epsilon(u,v)&=\ds\frac{1}{2}\ds\int_{\R^3}\Big(\epsilon^2|\nabla u|^2 +P(x)u^2 +
\epsilon^2|\nabla v|^2 +Q(x)v^2\Big)~dx \vspace{0.2cm}\\&\,\,\,\,\,\,-\ds\frac{1}{4}\ds\int_{\R^3}\Big(\mu_1|u|^4+\mu_2|v|^4\Big)~dx
-\frac{\beta}{2}\ds\int_{\R^3}u^2v^2~dx.
\end{array}
\end{equation}
Then $I_\epsilon \in C^2(H)$ and its critical points correspond to the  solutions of \eqref{1.1}.

Define$$Y_j:=\frac{\partial U_{x^j,\epsilon}}{\partial r}, Z_j:=\frac{\partial V_{x^j,\epsilon}}{\partial r},j=1,2,\cdots,2k,$$
where $x^j$ is defined in \eqref{1.3} and define

\begin{equation}\label{2.2'}
E=\Big\{(u,v)\in H:\sum\limits_{j=1}^{2k}\ds\int_{\R^3}(U_{x^j,\epsilon}^2Y_ju +V_{x^j,\epsilon}^2Z_jv)~dx=0\Big\}.
\end{equation}

Let $$J(\varphi,\psi)=I_\epsilon(U_r+\varphi, V_r+\psi),\,\,\,(\varphi,\psi)\in E.$$

Expand $J(\varphi,\psi)$ as follows:
\begin{equation}\label{2.2}
J(\varphi,\psi)=J(0,0)+l(\varphi,\psi)+\frac{1}{2}Q(\varphi,\psi)+R(\varphi,\psi),\,\,\,(\varphi,\psi)\in E,
\end{equation}
where
$$
\begin{array}{rl}
&l(\varphi,\psi)\\[5mm]
&=\sum\limits_{j=1}^{2k}(-1)^{j-1}\ds\int_{\R^3}(P(x)-1)U_{x^j,\epsilon}\varphi-
\mu_1\ds\int_{\R^3}\Big(U_r^3-\sum\limits_{j=1}^{2k}(-1)^{j-1}U^3_{x^j,\epsilon}\Big)\varphi\\[5mm]
&\quad+\sum\limits_{j=1}^{2k}(-1)^{j-1}\ds\int_{\R^3}(Q(x)-1)V_{x^j,\epsilon}\psi-
\mu_2\ds\int_{\R^3}\Big(V_r^3-\sum\limits_{j=1}^{2k}(-1)^{j-1}V^3_{x^j,\epsilon}\Big)\psi\\[5mm]
&\quad-\beta\ds\int_{\R^3}\Big(U_rV_r^2-\sum\limits_{j=1}^{2k}(-1)^{j-1}V_{x^j,\epsilon}^2U_{x^j,\epsilon}\Big)\varphi
-\beta\ds\int_{\R^3}\Big(U_r^2V_r-\sum\limits_{j=1}^{2k}(-1)^{j-1}V_{x^j,\epsilon}U_{x^j,\epsilon}^2\Big)\psi,
\end{array}$$

$$
\begin{array}{rl}
Q(\varphi,\psi)&=\ds\int_{\R^3}(\epsilon^2|\nabla \varphi|^2+P(x)\varphi^2-3\mu_1U_r^2\varphi^2)\\[5mm]
&\quad+\ds\int_{\R^3}(\epsilon^2|\nabla\psi|^2+Q(x)\psi^2-3\mu_2V_r^2\psi^2)\\[5mm]
&\quad-\beta\ds\int_{\R^3}(U_r^2\psi^2+4U_rV_r\varphi\psi+V_r^2\varphi^2)
\end{array}
$$
and
$$
\begin{array}{rl}
R(\varphi,\psi)=&\ds\int_{\R^3}\big(\mu_1U_r\varphi^3+\mu_2V_r\psi^3+\frac{\mu_1}{4}\varphi^4+\frac{\mu_2}{4}\psi^4\big)\\[5mm]
&-\ds\frac{\beta}{2}\ds\int_{\R^3}\big[(U_r+\varphi)^2(V_r+\psi)^2-U_r^2V_r^2-2(U_rV_r^2\varphi+U_r^2V_r\psi)\\[5mm]
&\quad\quad\quad\quad\,\,-(U_r^2\psi^2+V_r^2\varphi^2+4U_rV_r\varphi\psi)\big].
\end{array}
$$

In order to find a critical point $(\varphi,\psi)\in E$ for $J(\varphi,\psi)$, we need to discuss each term in the expansion \eqref{2.2}.

It is easy to check that
$$
\begin{array}{rl}
&\ds\int_{\R^3}(\epsilon^2\nabla u\nabla\varphi+P(x)u\varphi-3\mu_1U_r^2u\varphi)+\ds\int_{\R^3}(\epsilon^2\nabla v\nabla \psi+Q(x)v\psi-3\mu_2V_r^2v\psi)\\[5mm]
&-\beta\ds\int_{\R^3}(U_r^2v\psi+V^2_ru\varphi+2U_rV_ru\psi+2U_rV_rv\varphi)
\end{array}
$$
is a bounded bi-linear functional in $E$. Thus there exists a bounded linear operator $L$ from $E$ to $E$  such that
$$
\begin{array}{rl}
&\langle L(u,v),(\varphi,\psi) \rangle\\[5mm]
&=\ds\int_{\R^3}(\epsilon^2\nabla u\nabla\varphi+P(x)u\varphi-3\mu_1U_r^2u\varphi)+\ds\int_{\R^3}(\epsilon^2\nabla v\nabla \psi+Q(x)v\psi-3\mu_2V_r^2v\psi)\\[5mm]
&\quad-\beta\ds\int_{\R^3}(U_r^2v\psi+V^2_ru\varphi+2U_rV_ru\psi+2U_rV_rv\varphi),\,\,\,(u,v),(\varphi,\psi) \in E.
\end{array}
$$

From the above analysis, we have the following  lemma.

\begin{lemma}\label{lemma2.1} There is a constant $C>0$, independent of $\epsilon$, such that for any $r\in S_\epsilon$,
$$\|L(u,v)\|\leq C\|(u,v)\|_\epsilon,\,\,\,\,(u,v)\in E.$$
\end{lemma}
Next, we discuss the invertibility of $L$.

\begin{lemma}\label{lemma2.2} There exist constants $C_0>0$ and $\epsilon_0>0$, such that for any $0<\epsilon<\epsilon_0$
and any $r \in S_\epsilon$
$$\|L(u,v)\|\geq C_0\|(u,v)\|_\epsilon,\,\,\,\,(u,v)\in E.$$
\end{lemma}

\begin{proof}
We argue by contradiction.  Suppose that there exist $\epsilon_n \to 0^+, r_n \in S_{\epsilon_n}$ and $(u_n,v_n) \in E$ such that
$$\|L(u_n,v_n)\|=o_n(1)\|(u_n,v_n)\|_{\epsilon_n}.$$
Since $L$ is linear,
we may as well assume that $$\|(u_n,v_n)\|_{\epsilon_n}^2=\epsilon_n^3$$
and
\begin{equation}\label{2.3}
\|L(u_n,v_n)\|=o_n(1)\epsilon_n^\frac{3}{2}.
\end{equation}
Then
$$
\langle L(u_n,v_n),(\varphi,\psi)\rangle=o_n(1)\|(\varphi,\psi)\|_{\epsilon_n}\epsilon_n^\frac{3}{2},\,\,\,\forall(\varphi,\psi) \in E.
$$
That is,
\begin{equation}\label{2.5}
\begin{array}{rl}
&\ds\int_{\R^3}(\epsilon_n^2\nabla u_n\nabla \varphi+P(x)u_n\varphi-3\mu_1U_{r_n}^2u_n\varphi)+\ds\int_{\R^3}(\epsilon_n^2\nabla v_n\nabla\psi+Q(x)v_n\psi-3\mu_2V_{r_n}^2v_n\psi)\\[5mm]
&-\beta\ds\int_{\R^3}(U_{r_n}^2v_n\psi+V^2_{r_n}u_n\varphi+2U_{r_n}V_{r_n}u_n\psi+2U_{r_n}V_{r_n}v_n\varphi)\\[5mm]
&=o_n(1)\|(\varphi,\psi)\|_{\epsilon_n}\epsilon_n^\frac{3}{2},\,\,\,\forall(\varphi,\psi) \in E.
\end{array}
\end{equation}
In particular, we have
\begin{equation}\label{2.6} 
\begin{array}{rl}
&\ds\int_{\R^3}(\epsilon_n^2|\nabla u_n|^2+P(x)|u_n|^2-3\mu_1U_{r_n}^2u_n^2)+\ds\int_{\R^3}(\epsilon_n^2|\nabla v_n|^2+Q(x)|v_n|^2-3\mu_2V_{r_n}^2v_n^2)\\[5mm]
&-\beta\ds\int_{\R^3}(U_{r_n}^2v_n^2+V_{r_n}^2u_n^2+4U_{r_n}V_{r_n}u_nv_n)\\[5mm]
&=o_n(1)\epsilon_n^3.
\end{array}
\end{equation}
We set $\tilde{u}_n(y)=u_n(\epsilon_ny+x^1)$ and $\tilde{v}_n(y)=v_n(\epsilon_ny+x^1)$. Then
\begin{equation}\label{2.7}
\ds\int_{\R^3}(|\nabla \tilde{u}_n|^2+P(\epsilon_ny+x^1)\tilde{u}_n^2+|\nabla\tilde{v}_n|^2+Q(\epsilon_ny+x^1)\tilde{v}_n^2)=1.
\end{equation}
Therefore, there exist $u,v \in H^1(\R^3)$ such that $n \to \infty,$
$$ \tilde{u}_n \to u,~~ \hbox{weakly~in~} H^1_{loc}(\R^3),~~~~\,\,\tilde{u}_n\to u,~~~\hbox{strongly~in~}L_{loc}^2(\R^3),$$
$$ \tilde{v}_n \to v,~~ \hbox{weakly~in~} H^1_{loc}(\R^3),~~~~\,\,\tilde{v}_n\to v,~~~\hbox{strongly~in~}L_{loc}^2(\R^3).$$
Since $\tilde{u}_n$ and $\tilde{v}_n$ are even in $ y_2$ and $y_3$, it is easy to see that $u$ and $v$ are even in $ y_2$ and $y_3$.

On the other hand, from the definition of $E$, we know that $(u,v)$ satisfies
\begin{equation}\label{2.8}
\ds\int_{\R^3}\Big(U^2\frac{\partial U}{\partial x_1}u +V^2\frac{\partial V}{\partial x_1}v\Big)=0.
\end{equation}

Now we claim that $(u,v)$ satisfies
\begin{equation}\label{2.9}\left\{%
\begin{array}{ll}
-\Delta u +u -3\mu_1U^2u-\beta V^2u-2\beta UVv=0,\,\,&x\in \R^{3},\vspace{0.2cm}\\
-\Delta v +v -3\mu_2V^2v-\beta U^2v-2\beta UVu=0, \,\,&x\in \R^{3}.
\end{array}
\right.
\end{equation}
Define
$$\widehat{E}=\left\{ (\varphi,\psi) \in H^1(\R^3) \times H^1(\R^3):\ds\int_{\R^3}\Big(U^2\frac{\partial U}{\partial x_1}u +V^2\frac{\partial V}{\partial x_1}v\Big)=0 \right\}.$$

For any $R>0$, let $(\varphi,\psi) \in C_0^\infty(B_R(0)) \times C_0^\infty(B_R(0))\cap \widehat{E}$ and be even in $y_2$ and $y_3$. Then $(\varphi_n(y),\psi_n(y)):=(\varphi(\frac{y-x^1}{\epsilon_n}),\psi(\frac{y-x^1}{\epsilon_n})) \in C_0^\infty(B_{R\epsilon_n}(x^1))\times C_0^\infty(B_{R\epsilon_n}(x^1)).$
 Inserting $(\varphi_n(y),\psi_n(y))$ into \eqref{2.5}, we find that
\begin{equation}\label{2.10}\begin{array}{rl}
&\ds\int_{\R^3}(\nabla u\nabla\varphi+u\varphi-3\mu_1U^2u\varphi)+\ds\int_{\R^3}(\nabla v\nabla\psi+v\psi-3\mu_2V^2v\psi)\\[5mm]
&-\beta\ds\int_{\R^3}(U^2v\psi+V^2u\varphi+2UVu\psi+2UVv\varphi)=0.\\[5mm]
\end{array}
\end{equation}

However, since $u$ and $v$ are even in $y_2$ and $y_3$, \eqref{2.10} holds for any function $(\varphi,\psi) \in C_0^\infty(B_R(0)) \times C_0^\infty(B_R(0))$, which is odd in $y_2$ or $y_3$. Therefore, \eqref{2.10} holds for any $(\varphi,\psi) \in C_0^\infty(B_R(0)) \times C_0^\infty(B_R(0))\cap \widehat{E}$.
By the density of $C_0^\infty(B_R(0)) \times C_0^\infty(B_R(0))$ in $H^1(\R^3) \times H^1(\R^3)$, we obtain that
\begin{equation}\label{2.11}\begin{array}{rl}
&\ds\int_{\R^3}(\nabla u\nabla \varphi+u\varphi-3\mu_1U^2u\varphi)+\ds\int_{\R^3}(\nabla v\nabla\psi+v\psi-3\mu_2V^2v\psi)\\[5mm]
&-\beta\ds\int_{\R^3}(U^2v\psi+V^2u\varphi+2UVu\psi+2UVv\varphi)=0,\,\,\,\forall (\varphi,\psi) \in \widehat{E}.\\[5mm]
\end{array}
\end{equation}
Noting that $(U,V)=(\alpha w, \gamma w)$ and $w$ is a solution of \eqref{11.4}, we can show that \eqref{2.10} holds for
$(\varphi,\psi)=(\frac{\partial U}{\partial x_1},\frac{\partial V}{\partial x_1})$. Thus \eqref{2.10} is true for any
$(\varphi,\psi) \in H^1(\R^3) \times H^1(\R^3)$. Therefore, we have verified \eqref{2.9}.

From Proposition 2.3 of \cite{PW}, we can know that $(U,V)$ is non-degenerate.  Since we work in the space of functions
which are even in $y_2$ and $y_3$, the kernel of $(U,V)$ is given by the one dimensional
$(\theta(\beta)\frac{\partial w}{\partial x_1},\frac{\partial w}{\partial x_1})$. So, we get $(u,v)=c(\frac{\partial U}{\partial x_1},\frac{\partial V}{\partial x_1})$ for some constant $c$. From \eqref{2.8} we can see $(u,v)=(0,0)$.

As a result,
$$\ds\int_{B_R(-x^1)}(u_n^2+v_n^2)=o_n(1)\epsilon^3, \forall R>0.$$

By direct calculation, we get
$$\ds\int_{\R^3}\left(U_{r_n}^2u_n^2 +V_{r_n}^2v_n^2\right)=o_n(1)\epsilon_n^3+o_R(1)\epsilon_n^3.$$

As a result,
\begin{equation}\label{2.12}\begin{array}{rl}
&o_n(1)\epsilon_n^3\vspace{0.15cm}\\=&\ds\int_{\R^3}(\epsilon_n^2|\nabla u_n|^2+P(x)|u_n|^2-3\mu_1U_{r_n}^2u_n^2)+\ds\int_{\R^3}(\epsilon_n^2|\nabla v_n|^2+Q(x)|v_n|^2-3\mu_2V_{r_n}^2v_n^2)\\[5mm]
&-\beta\ds\int_{\R^3}(U_{r_n}^2v_n^2+V^2_{r_n}u_n^2+4U_{r_n}V_{r_n}u_nv_n)\\[5mm]
&=(1+o_n(1)+o_R(1))\epsilon_n^3.
\end{array}
\end{equation}
This is a contradiction. So we complete the proof.
\end{proof}

\begin{lemma}\label{lemma2.4}
For any $(\varphi,\psi) \in E,$  
we have
$$
\|R(\varphi,\psi)\|=O(\epsilon^{-\frac{3}{2}}\|(\varphi,\psi)\|_\epsilon^3
+\epsilon^{-4}\|(\varphi,\psi)\|_\epsilon^4),
$$
$$
\|R^\prime(\varphi,\psi)\|
=O(\epsilon^{-\frac{3}{2}}\|(\varphi,\psi)\|_\epsilon^2+\epsilon^{-4}\|(\varphi,\psi)\|_\epsilon^3)
$$
and
$$\|R^{\prime\prime}(\varphi,\psi)\|
=O(\epsilon^{-\frac{3}{2}}\|(\varphi,\psi)\|_\epsilon+\epsilon^{-4}\|(\varphi,\psi)\|_\epsilon^2).
$$
\end{lemma}

\begin{proof}
By direct calculation, we have, for any $(u_1,v_1),(u_2,v_2)\in E$
$$
\begin{array}{rl}
|R(\varphi,\psi)|&=\Big|\ds\int_{\R^3}(\mu_1U_r\varphi^3+\mu_2V_r\psi^3+\frac{\mu_1}{4}\varphi^4+\frac{\mu_2}{4}\psi^4)\\[5mm]
&\quad\,\,-\ds\frac{\beta}{2}\ds\int_{\R^3}
[(U_r+\varphi)^2(V_r+\psi)^2-U_r^2V_r^2-2(U_rV_r^2\varphi+U_r^2V_r\psi)\\[5mm]
&\quad\quad\quad\quad\quad-(U_r^2\psi^2+V_r^2\varphi^2+4U_rV_r\varphi\psi)]\Big|\\[5mm]
&=\Big|\ds\int_{\R^3}(\mu_1U_r\varphi^3+\mu_2V_r\psi^3+\frac{\mu_1}{4}\varphi^4+\frac{\mu_2}{4}\psi^4)\\[5mm]
&\quad\,\,-\ds\frac{\beta}{2}\ds\int_{\R^3}(\varphi^2\psi^2+2U_r\varphi\psi^2+2V_r\varphi^2\psi)\Big|\\[5mm]
&\leq C\ds\int_{\R^3}\Big(\sum\limits_{j=1}^{2k}U_{x^j,\epsilon}|\varphi|^3+\varphi^4+\sum\limits_{j=1}^{2k}V_{x^j,\epsilon}|\psi|^3+\psi^4\Big)\\[5mm]
&\leq C(\epsilon^{-\frac{3}{2}}\|(\varphi,\psi)\|_\epsilon^3+\epsilon^{-4}\|(\varphi,\psi)\|_\epsilon^4)
\end{array}
$$
and
$$
\begin{array}{rl}
&|\langle R^\prime(\varphi,\psi),(u_1,v_1)\rangle |\vspace{0.15cm}\\
&=\Big|\ds\int_{\R^3}(3\mu_1U_r\varphi^2u_1 +3\mu_2V_r\psi^2v_1+\mu_1\varphi^3u_1+\mu_2\psi^3v_1)\\[5mm]
&\quad\,\,+\beta\ds\int_{\R^3}(\varphi\psi^2u_1+\varphi^2\psi v_1+2U_r\varphi\psi v_1+2U_r\psi^2u_1+2V_r\varphi\psi u_1+2V_r\varphi^2v_1)\Big|\\[5mm]
&\leq C\ds\int_{\R^3}\Big[\big(\sum\limits_{j=1}^{2k}U_{x^j,\epsilon}+\sum\limits_{j=1}^{2k}V_{x^j,\epsilon}\big)(\varphi^2+\psi^2)(|u_1|+|v_1|)
+(|\varphi|^3+|\psi|^3)(|v_1|+|u_1|)\Big]\\[5mm]
&\leq C\big(\epsilon^{-\frac{3}{2}}\|(\varphi,\psi)\|_\epsilon^2+\epsilon^{-4}\|(\varphi,\psi)\|_\epsilon^3\big)\|(u_1,v_1)\|_\epsilon.\\[5mm]
\end{array}$$
And by similar calculation, we get that
$$\begin{array}{rl}
|\langle R^{\prime\prime}(\varphi,\psi)(u_1,v_1),(u_2,v_2)\rangle |&\leq C(\epsilon^{-\frac{3}{2}}\|(\varphi,\psi)\|_\epsilon+\epsilon^{-4}\|(\varphi,\psi)\|_\epsilon^2)\|(u_1,v_1)\|_\epsilon\|(u_2,v_2)\|_\epsilon.\\[5mm]
\end{array}$$
So we complete the proof of this lemma.
\end{proof}

\begin{lemma}\label{lemma2.5} There exists a small constant $\tau \in D$ such that
$$
\|l\|=O\big(r^{\min\{n,m\}}+e^{-\frac{3(1-\tau)r}{\epsilon}}
+e^{-\frac{2r\sin\frac{\pi}{2k}}{\epsilon}}\big)\epsilon^\frac{3}{2},
$$
where $D=\{x \in (0,\frac{1}{3})|(1-x)(2-x) \geq \frac{11\sqrt{2}}{10}\}.$
\end{lemma}

\begin{proof}
By direct computations, we have
\begin{equation}\label{2.14}
\begin{array}{rl}
&\sum\limits_{j=1}^{2k}(-1)^{j-1}\ds\int_{\R^3}(P(x)-1)U_{x^j,\epsilon}\varphi+\sum\limits_{j=1}^{2k}(-1)^{j-1}\ds\int_{\R^3}(Q(x)-1)V_{x^j,\epsilon}\psi\\[5mm]
&\leq \sum\limits_{j=1}^{2k}\Big(\ds\int_{\R^3}|(P(x)-1)|^2U^2_{x^j,\epsilon}\Big)^\frac{1}{2}\Big(\ds\int_{\R^3}\varphi^2\Big)^\frac{1}{2}+
\sum\limits_{j=1}^{2k}\Big(\ds\int_{\R^3}|(Q(x)-1)|^2V^2_{x^j,\epsilon}\Big)^\frac{1}{2}\Big(\ds\int_{\R^3}\psi^2\Big)^\frac{1}{2}\\[5mm]
&\leq C\epsilon^\frac{3}{2}(r^m+e^{-\frac{3(1-\tau)r}{\epsilon}})\|\varphi\|_{\epsilon,P}
+C\epsilon^\frac{3}{2}(r^n+e^{-\frac{3(1-\tau)r}{\epsilon}})\|\psi\|_{\epsilon,Q}\\[5mm]
&\leq C(r^{\min\{m,n\}}+e^{-\frac{3(1-\tau)r}{\epsilon}})\epsilon^\frac{3}{2}\|(\varphi,\psi)\|_\epsilon,\\[5mm]
\end{array}
\end{equation}
\begin{equation}\label{2.15}\begin{array}{rl}
&\mu_1\ds\int_{\R^3}\Big(\sum\limits_{j=1}^{2k}(-1)^{j-1}
U_{x^j,\epsilon}^3-U_r^3\Big)\varphi+\mu_2\ds\int_{\R^3}\Big(\sum\limits_{j=1}^{2k}(-1)^{j-1}
V_{x^j,\epsilon}^3-V_r^3\Big)\psi\\[5mm]
&\leq C\epsilon^\frac{3}{2}e^{-\frac{|x^1-x^2|}{\epsilon}}\|(\varphi,\psi)\|_\epsilon
\end{array}
\end{equation}
and
\begin{equation}\label{2.16}\begin{array}{rl}
&\beta\ds\int_{\R^3}\Big(\sum\limits_{j=1}^{2k}(-1)^{j-1}V_{x^j,\epsilon}^2U_{x^j,\epsilon}-V_r^2U_r\Big)\varphi
+\beta\ds\int_{\R^3}\Big(\sum\limits_{j=1}^{2k}(-1)^{j-1}U_{x^j,\epsilon}^2V_{x^j,\epsilon}-U_r^2V_r\Big)\psi\\[5mm]
&\leq C\epsilon^\frac{3}{2}e^{-\frac{|x^1-x^2|}{\epsilon}}\|(\varphi,\psi)\|_\epsilon.\\[5mm]
\end{array}
\end{equation}
Combining \eqref{2.14}, \eqref{2.15}  \eqref{2.16} and the definition of $l$, we can deduce that
$$
\|l\|=O\big(r^{\min\{n,m\}}+e^{-\frac{3(1-\tau)r}{\epsilon}}+e^{-\frac{2r\sin\frac{\pi}{2k}}{\epsilon}}\big)
\epsilon^\frac{3}{2}.
$$
\end{proof}

\begin{proposition}\label{pro2.6}  For $\epsilon$ sufficiently small, there exists a $C^1-$ map $(\varphi,\psi)$ from $S_\epsilon$ to $H$:
$(\varphi,\psi):=(\varphi(r),\psi(r)), r=|x|$ satisfying $(\varphi,\psi) \in E$ and
$$
\Big\langle\frac{\partial J(\varphi,\psi)}{\partial (\varphi,\psi)}, (g,h)\Big\rangle=0, \forall (g,h) \in E.
$$
Moreover, there exists a small constant $0<\tau_2< \min\{\frac{1}{5}, \frac{\min\{n, m\}-1-\sigma}{\min\{n, m\}}\}$ such that
$$
\|(\varphi,\psi)\|_\epsilon \leq \big(r^{(1-\tau_2)\min\{m,n\}}+e^{-\frac{3(1-\tau_2)(1-\tau)r}{\epsilon}}+e^{-\frac{(1-\tau_2)2r\sin\frac{\pi}{2k}}{\epsilon}}\big)\epsilon^\frac{3}{2}.
$$

\end{proposition}

\begin{proof}
It follows from Lemma \ref{lemma2.5} that $l$ is a bounded linear  functional in $E$. Thus there exists an $l^\prime \in E$ such that $l(\varphi,\psi)=\langle l^\prime,(\varphi,\psi)\rangle$. Thus finding a critical point for $J(\varphi,\psi)$ is equivalent to solving
\begin{equation}\label{2.17}
l^\prime + L(\varphi,\psi) +R^\prime(\varphi,\psi)=0.
\end{equation}
By Lemma \ref{lemma2.2}, $L$ is invertible. Hence \eqref{2.17} can be written as
\begin{equation}\label{2.18}
(\varphi,\psi)=A(\varphi,\psi):=-L^{-1}l^\prime-L^{-1}R^\prime(\varphi,\psi).
\end{equation}

We choose a small constant $0<\tau_2< \min\{\frac{1}{5}, \frac{\min\{n, m\}-1-\sigma}{\min\{n, m\}}\}$ and set
$$
S=\Big\{(\varphi,\psi) \in E:\|(\varphi,\psi)\|_\epsilon \leq \epsilon^\frac{3}{2}\big(r^{(1-\tau_2)\min\{m,n\}}+e^{-\frac{3(1-\tau_2)(1-\tau)r}{\epsilon}}
+e^{-\frac{(1-\tau_2)2r\sin\frac{\pi}{2k}}{\epsilon}}\big)\Big\}.
$$
For $\epsilon$ sufficiently small, we have
$$
\begin{array}{rl}
\|A(\varphi,\psi)\| &\leq C\|l^\prime \| +C\|R^\prime(\varphi,\psi)\|\\[5mm]
&\leq C\epsilon^\frac{3}{2}\big(r^{\min\{n,m\}}+e^{-\frac{3(1-\tau)r}{\epsilon}}+e^{-\frac{2r\sin\frac{\pi}{2k}}{\epsilon}}\big)
\\[5mm]
&\,\,\,\,\,+C(\epsilon^{-\frac{3}{2}}\|(\varphi,\psi)\|_\epsilon^2+\epsilon^{-4}\|(\varphi,\psi)\|_\epsilon^3)\\[5mm]
&\leq \epsilon^\frac{3}{2}\big(r^{(1-\tau_2)\min\{m,n\}}+e^{-\frac{3(1-\tau_2)(1-\tau)r}{\epsilon}}
+e^{-\frac{(1-\tau_2)2r\sin\frac{\pi}{2k}}{\epsilon}}\big),\,\,\,
\forall (\varphi,\psi) \in S,
\end{array}$$
which implies that $A$ is a map from $S$ to $S$.

On the other hand, for $\epsilon$ sufficiently small, we have
$$
\begin{array}{rl}
&|A(\varphi_1,\psi_1)-A(\varphi_2,\psi_2)|\\[5mm]
&\leq C|R^\prime(\varphi_1,\psi_1)-R^\prime(\varphi_2,\psi_2)|\\[5mm]
&\leq C\|R^{\prime\prime}(\lambda(\varphi_1,\psi_1)+(1-\lambda)(\varphi_2,\psi_2))\|\|(\varphi_1,\psi_1)-(\varphi_2,\psi_2))\|_\epsilon\\[5mm]
&\leq \ds\frac{1}{2}\|(\varphi_1,\psi_1)-(\varphi_2,\psi_2))\|_\epsilon.
\end{array}$$
Thus for $\epsilon$ sufficiently small, $A$ is a contraction map. Therefore we have proved that when $\epsilon$ is sufficiently small,
$A$ is a contraction map from $S$ to $S$. So  the results   follow from the  contraction mapping theorem. This completes the proof.

\end{proof}

Now we are ready to prove Theorem \ref{Th1}. Let $(\varphi_r,\psi_r)=(\varphi(r),\psi(r))$ be the map obtained in Proposition \ref{pro2.6}.
Define
$$F(r)=I_\epsilon(U_r+\varphi_r,V_r+\psi_r),\,\,\,\, r \in S_\epsilon.$$

With the same argument as in \cite{CNY,R}, we can easily check that if $r$ is a critical point of $F(r)$, then $(U_r+\varphi_r,V_r+\psi_r)$ is
a critical point of $I_\epsilon$.\\

{\textbf{Proof of Theorem \ref{Th3}}}\,\,\, It follows from Lemmas \ref{lemma2.1} and \ref{lemma2.4} that
$$
\|L(\varphi_r,\psi_r)\| \leq C\|(\varphi_r,\psi_r)\|_\epsilon,\,\,\,\, \|R(\varphi,\psi)\|\leq C\big(\epsilon^{-\frac{3}{2}}\|(\varphi,\psi)\|_\epsilon^3+\epsilon^{-4}\|(\varphi,\psi)\|_\epsilon^4\big).
$$
So from Lemma \ref{lemma2.5} and Proposition \ref{proA2}, we obtain that
$$
\begin{array}{rl}
F(r)&
=2k\epsilon^3\Big[A+aBr^m+bC_0r^n+C(\frac{\mu_1\alpha^4}{2}+\frac{\mu_2\gamma^4}{2}
+\beta\alpha^2\gamma^2)e^{-\frac{2r\sin\frac{\pi}{2k}}{\epsilon}}+O(r^{\min\{m-1,n-1\}}\epsilon)\Big].\\[5mm]
\end{array}
$$
Without loss of generality, we may as well assume that $n>m$. Therefore
$$F(r)=2k\epsilon^3\Big[A+aBr^m+Ce^{-\frac{2r\sin\frac{\pi}{2k}}{\epsilon}}+O(r^{m-1}\epsilon)\Big],$$
where $A,B,C$ are fixed positive constant.

Consider $\min\{F(r):r\in S_\epsilon\}$, where $S_\epsilon$ is defined in \eqref{s}.

Let $$f(r):= aBr^m +
Ce^{-\frac{2rsin\frac{\pi}{2k}}{\epsilon}}.$$

 By the assumption, we know that $a>0$. So by direct calculation, we can  get that $f(r)$ has a local minimum
point
$$\bar{r}=\frac{m+o_\epsilon(1)}{2\sin\frac{\pi}{2k}}\epsilon\ln\frac{1}{\epsilon}.$$
So there exists $\epsilon_0>0$ such that for any $\epsilon \in
(0,\epsilon_0]$, there is $r_0 \in S_\epsilon$ such that
$f^\prime(r_0)=0$.

By direct computation, we can obtain that
\begin{eqnarray*}
F(\bar{r})&=&2k\epsilon^3\Big[A+\Big(\frac{m+o_\epsilon(1)}{2\sin\frac{\pi}{2k}}\Big)^maB\Big(\epsilon\ln\ds\frac{1}{\epsilon}\Big)^m
+\frac{m aB}{2\sin\frac{\pi}{2k}}r^{m-1}\epsilon+O\left(r^{m-1}\epsilon\right)\Big]\\
&=&2k\epsilon^3\Big[A+\Big(aB\Big(\frac{m}{2\sin\frac{\pi}{2k}}\Big)^m+o_\epsilon(1)\Big)\Big(\epsilon\ln\ds\frac{1}{\epsilon}\Big)^m\Big].
\end{eqnarray*}
On the other hand, we also have
\begin{eqnarray*}
F\Big(\frac{m-\delta}{2\sin\frac{\pi}{2k}}\epsilon\ln\frac{1}{\epsilon}\Big)
&=&2k\epsilon^3\Big[A+aB\Big(\frac{m-\delta}{2\sin\frac{\pi}{2k}}\Big)^m\Big(\epsilon\ln\ds\frac{1}{\epsilon}\Big)^m
+C\epsilon^{m-\delta}+O(r^{m-1}\epsilon)\Big]\\
&\geq&2k\epsilon^3(A+C\epsilon^{m-\delta})
\end{eqnarray*}
and
\begin{eqnarray*}
F\Big(\frac{m+\delta}{2\sin\frac{\pi}{2k}}\epsilon\ln\frac{1}{\epsilon}\Big)&=&2k\epsilon^3\Big[A+aB\Big(\frac{m+\delta}{2\sin\frac{\pi}{2k}}\Big)^m\Big(
\epsilon\ln\frac{1}{\epsilon}\Big)^m
+C\epsilon^{m+\delta}+O(r^{m-1}\epsilon)\Big]\\
&=&2k\epsilon^3\Big[A+\Big(aB\Big(\frac{m+\delta}{2\sin\frac{\pi}{2k}}\Big)^m+o_\epsilon(1)\Big)\Big(\epsilon\ln\ds\frac{1}{\epsilon}\Big)^m\Big].
\end{eqnarray*}
Hence, $F(r)$  has a local minimum point $r_\epsilon$  in
$S_\epsilon$, and $ r_\epsilon$ is an interior point of
$S_\epsilon$. Thus $r_\epsilon$ is a critical point of $F(r)$.
As a result, $(U_{r_\epsilon}+\varphi_{r_\epsilon},V_{r_\epsilon}+\psi_{r_\epsilon}) $ is a solution
of \eqref{1.1}.

For the case $m>n,$  the same method can be used to prove the result.

For the  case $m=n$, then
$$F(r)=2k\epsilon^3\Big[A+(aB+bC_0)r^m+Ce^{-\frac{2rsin\frac{\pi}{2k}}{\epsilon}}+O(r^{m-1}\epsilon)\Big].$$
And let
$$f(r)=(aB+bC_0)r^m+Ce^{-\frac{2rsin\frac{\pi}{2k}}{\epsilon}}.$$
Using the above methods, we can prove the result.
This completes the proof.

\section{Segregated Vector Solutions and the proof of Theorem \ref{Th2}}\label{s3}
\def\theequation{3.\arabic{equation}}\makeatother
\setcounter{equation}{0}

In this section we consider segregated vector solutions and prove Theorem \ref{Th2} by proving Theorem \ref{Th4}. Let
$$\widetilde{Y}_j=\frac{\partial U_{1,x^j,\epsilon}}{\partial r},\,\,\,\widetilde{Z}_j=\frac{\partial U_{2,y^j,\epsilon}}{\partial \rho},\,\,\,j=1,2,\cdots,2k,$$
where $x^j$ and $y^j$ are defined in \eqref{1.3}, \eqref{1.4} respectively.

For simplicity of notation, in the sequel we use $U_{1,x^j,\epsilon}$ and $U_{2,y^j,\epsilon}$ to replace $U_{x^j,\epsilon}$ and
$V_{x^j,\epsilon}$ respectively. In this section, we assume

\begin{equation}\label{3.1}
(r,\rho) \in S_\epsilon \times S_\epsilon.
\end{equation}

Define
\begin{equation}\label{3.2}
\tilde{E}=\Big\{(\varphi,\psi)\in H: \sum\limits_{j=1}^{2k}\int_{\R^3}U^2_{1,x^j,\epsilon}\widetilde{Y}_j\varphi=0,\,\,\sum\limits_{j=1}^{2k}\int_{\R^3}U^2_{2,y^j,\epsilon}\widetilde{Z}_j\psi=0\Big\}.
\end{equation}

Let
$$\tilde{J}(\tilde{\varphi},\tilde{\psi})=I_\epsilon(\tilde{U}_r+\tilde{\varphi},\tilde{V}_\rho+\tilde{\psi}),\,\,\,(\tilde{\varphi},\tilde{\psi})
\in \tilde{E}.$$

Then, similar to \eqref{2.2}, $\tilde{J}(\tilde{\varphi},\tilde{\psi})$ has the following expansion:
$$\tilde{J}(\tilde{\varphi},\tilde{\psi})=\tilde{J}(0,0)+\tilde{l}(\tilde{\varphi},\tilde{\psi})+\frac{1}{2}\widetilde{Q}
(\tilde{\varphi},\tilde{\psi})+\tilde{R}(\tilde{\varphi},\tilde{\psi}),\,\,(\tilde{\varphi},\tilde{\psi})\in \tilde{E},$$
where $\widetilde{Q}
(\tilde{\varphi},\tilde{\psi})$ and $\tilde{R}(\tilde{\varphi},\tilde{\psi})$ are the same as $Q(\varphi,\psi)$ and $R(\varphi,\psi)$
in section 2 if $U_{x^j,\epsilon},V_{x^j,\epsilon},\varphi,$ and $\psi$ are replaced by $U_{1,x^j,\epsilon},U_{2,y^j,\epsilon},\tilde{\varphi},$ and $\widetilde{\psi}$ respectively. We note that there exists a bounded linear
operator $\tilde{B}_\epsilon:\tilde{E}\to\tilde{E}$ corresponding to $\tilde{Q}
(\tilde{\varphi},\tilde{\psi})$.

Note that $\tilde{l}(\tilde{\varphi},\tilde{\psi})$ has the following form
$$
\begin{array}{rl}
&\tilde{l}(\tilde{\varphi},\tilde{\psi})\\[5mm]
&=\ds\sum\limits_{j=1}^{2k}(-1)^{j-1}\ds\int_{\R^3}(P(|x|)-1)U_{1,x^j,\epsilon}\tilde{\varphi}
-\mu_1\ds\int_{\R^3}\Big(\tilde{U}_r^3-\sum\limits_{j=1}^{2k}(-1)^{j-1}U^3_{1,x^j,\epsilon}\Big)\tilde{\varphi}\\[5mm]
&\quad+\ds\sum\limits_{j=1}^{2k}(-1)^{j-1}\ds\int_{\R^3}(Q(|x|)-1)U_{2,y^j,\epsilon}\tilde{\psi}
-\mu_2\ds\int_{\R^3}\Big(\tilde{V_\rho}^3-\sum\limits_{j=1}^{2k}(-1)^{j-1}U^3_{2,y^j,\epsilon}\Big)\tilde{\psi}\\[5mm]
&\quad-\beta\ds\int_{\R^3}(\tilde{U}_r\tilde{V}_\rho^2\tilde{\varphi}+\tilde{U}_r^2\tilde{V}_\rho\tilde{\psi}).
\end{array}$$

From the above analysis, we have the following lemma:
\begin{lemma}\label{lm3.1}
There exists a constant $C>0$,independent of $\epsilon$, such that for any $(r,\rho) \in S_\epsilon \times S_\epsilon$
$$\|\tilde{B_\epsilon}(\varphi,\psi)\| \leq C \|(\varphi,\psi)\|_\epsilon,\,\,\,\,(\varphi,\psi)\in\tilde{E}.$$

\end{lemma}

\begin{lemma}\label{lm3.2}
There exist $\epsilon_0 >0,\beta_0>0$ and $C_0>0$ such that for any $\beta<\beta_0$ and  any $\epsilon \in(0,\epsilon_0), (r,\rho)\in S_\epsilon \times S_\epsilon,$
we have
$$\|\tilde{B_\epsilon}(\varphi,\psi)\| \geq C_0 \|(\varphi,\psi)\|_\epsilon,\,\,\,\,(\varphi,\psi)\in\tilde{E}.$$
\end{lemma}

\begin{proof}
The argument is similar to Lemma \ref{lemma2.2}. We argue by contradiction. Suppose that there are $\epsilon_n \to 0^+, (r_n,\rho_n)\in S_{\epsilon_n} \times S_{\epsilon_n}$
and $(\varphi_n,\psi_n) \in \tilde{E}$ with $\|(\varphi_n,\psi_n)\|^2_{\epsilon_n}=\epsilon_n^3$ satisfying
\begin{equation}\label{3.3}
\langle \tilde{B}_\epsilon(\varphi_n,\psi_n),(g,h)\rangle=o_n(1)\|(\varphi_n,\psi_n)\|_{\epsilon_n}\|(g,h)\|_{\epsilon_n},\,\,\,\forall (g,h)\in\tilde{E}.
\end{equation}
That is,
\begin{equation}\label{3.4}
\begin{split}
\begin{array}{rl}
&\ds\int_{\R^3}(\epsilon^2_n\nabla\varphi_n\nabla g +P(x)\varphi_n g-3\mu_1\tilde{U}_r^2\varphi_n g)\\[5mm]
&+\ds\int_{\R^3}(\epsilon^2_n\nabla\psi_n\nabla h +Q(x)\psi_n h-3\mu_2\tilde{V}_\rho^2\psi_n h)\\[5mm]
&-\beta\ds\int_{\R^3}(\tilde{U}_r^2\psi_n h+\tilde{V}_\rho^2\varphi_n g+2\tilde{U}_r\tilde{V}_\rho\varphi_n h+2\tilde{U}_r\tilde{V}_\rho\psi_n g)\\[5mm]
&=o_n(1)\|(\varphi_n,\psi_n)\|_{\epsilon_n}\|(g,h)\|_{\epsilon_n},,\,\,\,\forall (g,h)\in\tilde{E}.
\end{array}
\end{split}
\end{equation}
In particular, we have
\begin{equation}\label{3.5}
\begin{split}
\begin{array}{rl}
&\ds\int_{\R^3}(\epsilon^2_n|\nabla\varphi_n|^2 +P(x)|\varphi_n|^2-3\mu_1\tilde{U}_r^2\varphi_n^2)\\[5mm]
&+\ds\int_{\R^3}(\epsilon^2_n|\nabla\psi_n|^2 +Q(x)\psi_n^2-3\mu_2\tilde{V}_\rho^2\psi_n^2)\\[5mm]
&-\beta\ds\int_{\R^3}(\tilde{U}_r^2\psi_n^2+\tilde{V}_\rho^2\varphi_n^2+4\tilde{U}_r\tilde{V}_\rho\varphi_n \psi_n )\\[5mm]
&=o_n(1)\epsilon^3_n
\end{array}
\end{split}
\end{equation}
and
$$\ds\int_{\R^3}(\epsilon^2_n|\nabla\varphi_n|^2 +P(x)|\varphi_n|^2+\epsilon^2_n|\nabla\psi_n|^2 +Q(x)\psi_n^2)=\epsilon^3_n.$$

We set $\tilde{u}_n(x)=\varphi_n(\epsilon_n x+x^1),\,\,\,\tilde{v}_n(x)=\psi_n(\epsilon_n x+y^1)$. Then we have
$$\ds\int_{\R^3}(|\nabla\tilde{u}_n(x)|^2 +P(\epsilon_n x+x^1)|\tilde{u}_n(x)|^2+|\nabla\tilde{v}_n(x)|^2 +Q(\epsilon_n x+y^1)|\tilde{v}_n(x)|^2)=1.$$
Upon passing to a subsequence, we may as well assume that there exist $u,v\in H^1(\R^3)$ such that as $n \to +\infty$
$$\tilde{u}_n(x) \to u~~~\hbox{weakly in}~H^1_{loc}(\R^3),~~~\tilde{u}_n(x) \to u~~~\hbox{strongly in}~L^2_{loc}(\R^3),$$
$$\tilde{v}_n(x) \to v~~~\hbox{weakly in}~H^1_{loc}(\R^3),~~~\tilde{v}_n(x) \to v~~~\hbox{strongly in}~L^2_{loc}(\R^3).$$
Moreover, $u$ and $v$ satisfy
$$
\ds\int_{\R^3}\Big(\nabla\frac{\partial U_1}{\partial x_1}\nabla u+\frac{\partial U_1}{\partial x_1} u\Big)=0,
\,\,\,\ds\int_{\R^3}\Big(\nabla\frac{\partial U_2}{\partial x_1}\nabla v+\frac{\partial U_2}{\partial x_1} v\Big)=0.
$$

We claim that $u$ and $v$ satisfy
$$-\Delta u +u -3\mu_1U_1^2u=0,\,\,\,-\Delta v +v -3\mu_2U_2^2v=0.$$

Let $\tilde{\varphi}(x)\in C^\infty_0(B_R(0))$ and be even in $y_2$ and $y_3$. Define $ \tilde{\varphi}_n(x):=\tilde{\varphi}(\frac{ x-x^1}{\epsilon_n})
\in C^\infty_0(B_{\epsilon_n R}(x^1))$. Then inserting $( \tilde{\varphi}_n(x),0)$ into \eqref{3.4} and preceding as we have done in Lemma \ref{lemma2.2}, we can see that $u$ satisfies
$$-\Delta u +u -3\mu_1U_1^2u=0~~~~~\hbox{in}~~\R^3.$$
Also, by the non-degeneracy of $U_1$, we find that $u=0$.
In the same way, we also find that $v=0$.

As a result,
$$\ds\int_{B_R(-x^1)}\varphi_n^2=o_n(1)\epsilon^3_n,\,\,\,\,\ds\int_{B_R(-y^1)}\psi_n^2=o_n(1)\epsilon^3_n,\,\,\,\forall R>0.$$

Thus, it follows from \eqref{3.5} and Lemma \ref{LemmaA3} that
\begin{equation}\label{3.6}
\begin{split}
\begin{array}{rl}
o_n(1)\epsilon^3_n&=\ds\int_{\R^3}(\epsilon^2_n|\nabla\varphi_n|^2 +P(x)|\varphi_n|^2-3\mu_1\tilde{U}_r^2\varphi_n^2)\\[5mm]
&\,\,\,\,\,\,\,\,+\ds\int_{\R^3}(\epsilon^2_n|\nabla\psi_n|^2 +Q(x)\psi_n^2-3\mu_2\tilde{V}_\rho^2\psi_n^2)\\[5mm]
&\,\,\,\,\,\,\,\,-\beta\ds\int_{\R^3}(\tilde{U}_r^2\psi_n^2+\tilde{V}_\rho^2\varphi_n^2+4\tilde{U}_r\tilde{V}_\rho\varphi_n \psi_n )\\[5mm]
&\geq \|(\tilde{\varphi}_n,\tilde{\psi}_n)\|^2_{\epsilon_n}-C\beta\|(\tilde{\varphi}_n,\tilde{\psi}_n)\|^2_{\epsilon_n}
+\epsilon^3_n(o_n(1)+o_R(1)).
\end{array}
\end{split}
\end{equation}
If $\beta <\beta_0:=\frac{1}{C}$, and for large $n$ and large $R$, we get a contradiction. So the result in this Lemma is true. This completes the proof.
\end{proof}

From \eqref{2.14}, \eqref{2.15} and Lemma \ref{LemmaA3}, we can get the following Lemma.
\begin{lemma}\label{lm3.3}
There exists a small constant $\tilde{\tau}_1 \in D$ such that
$$
\begin{array}{rl}
\|\tilde{l}\|&=O\Big(r^m+\rho^n+e^{-\frac{3(1-\tilde{\tau}_1)r}{\epsilon}}+e^{-\frac{3(1-\tilde{\tau}_1)\rho}{\epsilon}}
+e^{-\frac{2r\sin\frac{\pi}{2k}}{\epsilon}}\\[5mm]
&\quad\quad+e^{-\frac{2\rho\sin\frac{\pi}{2k}}{\epsilon}}
+\frac{\beta}{(\ln\frac{1}{\epsilon})^\frac{1}{6}}
e^{-\frac{\sqrt{(\rho-r\cos\frac{\pi}{2k})^2+\epsilon^\frac{3}{2}(r\sin\frac{\pi}{2k})^2}}{\epsilon}}\Big),
\end{array}
$$
where $D$ has been defined in Lemma \ref{lemma2.5}.
\end{lemma}

\begin{proposition}\label{pro3.4}
For $\epsilon>0$ sufficiently small, There exists a $C^1$-map $(\tilde{\varphi},\tilde{\psi})$ from $S_\epsilon \times S_\epsilon$ to $H$:$(\tilde{\varphi},\tilde{\psi})=(\tilde{\varphi}(r,\rho),\tilde{\psi}(r,\rho)), r=|x^1|,\rho=|y^1|,$ satisfying
$(\tilde{\varphi},\tilde{\psi})\in\tilde{E},$ and
$$
\Big\langle\frac{\partial \tilde{J}(\tilde{\varphi},\tilde{\psi})}{\partial (\tilde{\varphi},\tilde{\psi})},(g,h)\Big\rangle=0,\,\,\,\forall (g,h)\in\tilde{E}
$$
Moreover,  there exists a small constant $0<\widetilde{\tau_2}< \min\{\frac{1}{5}, \frac{\min\{n, m\}-1-\sigma}{\min\{n, m\}}\}$ and a constant $\tilde{C}$ such that
$$
\begin{array}{rl}
\|(\tilde{\varphi},\tilde{\psi})\|_\epsilon& \leq \epsilon^\frac{3}{2} \Big(r^{(1-\tilde{\tau}_2)m}+\rho^{(1-\tilde{\tau}_2)n}+e^{-\frac{3(1-\tilde{\tau}_2)(1-\tilde{\tau}_1)r}{\epsilon}}
+e^{-\frac{3(1-\tilde{\tau}_2)(1-\tilde{\tau}_1)\rho}{\epsilon}}\\[5mm]
&\quad\quad\,\,+e^{-\frac{(1-\tilde{\tau}_2)2r\sin\frac{\pi}{2k}}{\epsilon}}+e^{-\frac{(1-\tilde{\tau}_2)2\rho\sin\frac{\pi}{2k}}{\epsilon}}
+\tilde{C}\frac{\beta}{(\ln\frac{1}{\epsilon})^\frac{1}{6}}
e^{-\frac{\sqrt{(\rho-r\cos\frac{\pi}{2k})^2+(r\sin\frac{\pi}{2k})^2}}{\epsilon}}\Big).
\end{array}
$$
\end{proposition}

\begin{proof}
From the definition of $\tilde{l}(\tilde{\varphi},\tilde{\psi})$, we know that $\tilde{l}(\tilde{\varphi},\tilde{\psi})$ is a bounded
linear functional in $\tilde{E}$. Thus it follows
from Reisz  Representation Theorem that there is an $\tilde{l}^\prime \in \tilde{E}$ such
that$$\tilde{l}(\tilde{\varphi},\tilde{\psi})=\langle\tilde{l}^\prime,(\tilde{\varphi},\tilde{\psi})\rangle.$$
So finding a critical point of $\tilde{J}(\tilde{\varphi},\tilde{\psi})$ is equivalent to solving
\begin{equation}\label{3.7}
\tilde{l}^\prime+\tilde{B}_\epsilon(\tilde{\varphi},\tilde{\psi})+\tilde{R}^\prime(\tilde{\varphi},\tilde{\psi})=0.
\end{equation}

By Lemma \ref{lm3.2}, $\tilde{B}_\epsilon$ is invertible. Hence \eqref{3.7} can be written as
$$(\tilde{\varphi},\tilde{\psi})=\tilde{A}(\tilde{\varphi},\tilde{\psi}):
=-\tilde{B}_\epsilon^{-1}\tilde{l}^\prime-\tilde{B}_\epsilon^{-1}\tilde{R}^\prime(\tilde{\varphi},\tilde{\psi}).$$

We choose a small constant $0<\widetilde{\tau_2}< \min\{\frac{1}{5}, \frac{\min\{n, m\}-1-\sigma}{\min\{n, m\}}\}$ and a sufficiently large constant $\tilde{C}$, and set
$$\begin{array}{rl}
\tilde{S}&=\bigg\{(\tilde{\varphi},\tilde{\psi}) \in \tilde{E}:\|(\tilde{\varphi},\tilde{\psi})\|_\epsilon \leq \epsilon^\frac{3}{2}\Big(r^{(1-\tilde{\tau}_2)m}+\rho^{(1-\tilde{\tau}_2)n}
+e^{-\frac{3(1-\tilde{\tau}_2)(1-\tilde{\tau}_1)r}{\epsilon}}\\[5mm]
&\quad\quad\quad\quad\quad\quad\quad\quad\quad\quad\quad\quad\quad
+e^{-\frac{3(1-\tilde{\tau}_2)(1-\tilde{\tau}_1)\rho}{\epsilon}}
+e^{-\frac{(1-\tilde{\tau}_2)2r\sin\frac{\pi}{2k}}{\epsilon}}+e^{-\frac{(1-\tilde{\tau}_2)2\rho\sin\frac{\pi}{2k}}{\epsilon}}
\\[5mm]
&\quad\quad\quad\quad\quad\quad\quad\quad\quad\quad\quad\quad\quad
+\tilde{C}\frac{\beta}{(\ln\frac{1}{\epsilon})^\frac{1}{6}}
e^{-\frac{\sqrt{(\rho-r\cos\frac{\pi}{2k})^2+(r\sin\frac{\pi}{2k})^2}}{\epsilon}}\Big)\bigg\}.\\
\end{array}$$
For $\epsilon$ sufficiently small, we have
$$\begin{array}{rl}
&\|\tilde{A}(\tilde{\varphi},\tilde{\psi})\|\\[5mm]
&\leq C\|\widetilde{l_k}\| +C\|\tilde{R}^\prime(\tilde{\varphi},\tilde{\psi)}\|\\[5mm]
&\leq C\Big(r^m+\rho^n+e^{-\frac{3(1-\tilde{\tau}_1)r}{\epsilon}}+e^{-\frac{3(1-\tilde{\tau}_1)\rho}{\epsilon}}
+e^{-\frac{2r\sin\frac{\pi}{2k}}{\epsilon}}+e^{-\frac{2\rho\sin\frac{\pi}{2k}}{\epsilon}}
\\[5mm]
&\quad\quad\,\,+\frac{\beta}{(\ln\frac{1}{\epsilon})^\frac{1}{6}}e^{-\frac{\sqrt{(\rho-r\cos\frac{\pi}{2k})^2+(r\sin\frac{\pi}{2k})^2}}{\epsilon}}\Big)\epsilon^\frac{3}{2}
+C(\epsilon^{-\frac{3}{2}}\|(\tilde{\varphi},\tilde{\psi})\|_\epsilon^2+\epsilon^{-4}\|(\tilde{\varphi},\tilde{\psi})\|_\epsilon^3)\\[5mm]
&\leq \Big(r^{(1-\tilde{\tau}_2)m}+\rho^{(1-\tilde{\tau}_2)n}+e^{-\frac{3(1-\tilde{\tau}_2)(1-\tilde{\tau}_1)r}{\epsilon}}
+e^{-\frac{3(1-\tilde{\tau}_2)(1-\tilde{\tau}_1)\rho}{\epsilon}}\\[5mm]
&\quad\quad+e^{-\frac{(1-\tilde{\tau}_2)2r\sin\frac{\pi}{2k}}{\epsilon}}+e^{-\frac{(1-\tilde{\tau}_2)2\rho\sin\frac{\pi}{2k}}{\epsilon}}
+\tilde{C}\frac{\beta}{(\ln\frac{1}{\epsilon})^\frac{1}{6}}e^{-\frac{\sqrt{(\rho-r\cos\frac{\pi}{2k})^2+(r\sin\frac{\pi}{2k})^2}}{\epsilon}}\Big)\epsilon^\frac{3}{2},\,\,\,
\forall (\tilde{\varphi},\tilde{\psi}) \in \tilde{S},
\end{array}$$
which implies that $\tilde{A}$ is a map from $\tilde{S}$ to $\tilde{S}$.

On the other hand, for $\epsilon$ sufficiently small, we get
$$
\begin{array}{rl}
&|\tilde{A}(\tilde{\varphi}_1,\tilde{\psi}_1)-\tilde{A}(\tilde{\varphi}_2,\tilde{\psi}_2)|\\[5mm]
&\leq C|\tilde{R}^\prime(\tilde{\varphi}_1,\tilde{\psi}_1)-\tilde{R}^\prime(\tilde{\varphi}_2,\tilde{\psi}_2)|\\[5mm]
&\leq C\|\widetilde{R}^{\prime\prime}(\lambda(\tilde{\varphi}_1,\tilde{\psi}_1)+(1-\lambda)(\tilde{\varphi}_2,\tilde{\psi}_2))\|\|(\tilde{\varphi}_1,\tilde{\psi}_1)-(\tilde{\varphi}_2,\tilde{\psi}_2))\|_\epsilon\\[5mm]
&\leq C\big[\epsilon^{-\frac{3}{2}}(\|(\tilde{\varphi}_1,\tilde{\psi}_1)\|_\epsilon+\|(\tilde{\varphi}_2,\tilde{\psi}_2)\|_\epsilon)
+\epsilon^{-4}(\|(\tilde{\varphi}_1,\tilde{\psi}_1)\|^2_\epsilon
+\|(\tilde{\varphi}_2,\tilde{\psi}_2)\|^2_\epsilon)
\big]\|(\tilde{\varphi}_1,\tilde{\psi}_1)-(\tilde{\varphi}_2,\tilde{\psi}_2))\|_\epsilon\\[5mm]
&\leq \ds\frac{1}{2}\|(\tilde{\varphi}_1,\tilde{\psi}_1)-(\tilde{\varphi}_2,\tilde{\psi}_2))\|_\epsilon.
\end{array}
$$
Thus for $\epsilon$ sufficiently small, $\tilde{A}$ is a contraction map. Therefore we have proved that when $\epsilon$ is sufficiently small,
$\tilde{A}$ is a contraction map from $\tilde{S}$ to $\tilde{S}$. So  the results   follow from the  contraction mapping theorem. This completes the proof.\\
\end{proof}

Now we are ready to prove Theorem \ref{Th2}. Let $(\tilde{\varphi}(r,\rho),\tilde{\psi}(r,\rho))$ be the map obtained in Proposition \ref{pro3.4}.
Define
$$\tilde{F}(r,\rho)=I_\epsilon(\tilde{U}_r+\tilde{\varphi}(r,\rho),\tilde{V}_\rho+\tilde{\psi}(r,\rho)),\,\,\, (r,\rho) \in S_\epsilon \times S_\epsilon.$$

We can check that for $\epsilon$ sufficiently small, if $(r,\rho)$ is a critical point of $\tilde{F}(r,\rho)$, then $(\tilde{U}_r+\tilde{\varphi}(r,\rho),\tilde{V}_\rho+\tilde{\psi}(r,\rho))$ is a critical point of $I_\epsilon$.\\

[\textbf{Proof of Theorem \ref{Th4}}]\,\,
 From Lemma \ref{lemma2.4} , \ref{lm3.3}, and Proposition \ref{pro3.4} , \ref{proA5}, we have
$$
\begin{array}{rl}
\tilde{F}(r,\rho)
&=2k\epsilon^3\Big[\tilde{A}+a\tilde{B}r^m+b\tilde{C}\rho^n+B_1e^{-\frac{2r\sin\frac{\pi}{2k}}{\epsilon}}
+B_2e^{-\frac{2\rho\sin\frac{\pi}{2k}}{\epsilon}}\\
&\quad \quad\quad\,\, 
+o_\epsilon(1)e^{-\frac{2\sqrt{(\rho-r\cos\frac{\pi}{k})^2+(r\sin\frac{\pi}{2k})^2}}{\epsilon}}+
O(r^{m-1}\epsilon+\rho^{n-1}\epsilon)\Big].\\
\end{array}
$$

Consider the minimization problem
$$\min\{\tilde{F}(r,\rho):(r,\rho) \in S_\epsilon\times S_\epsilon\}.$$

Since $\tilde{F}(r,\rho)$ is defined in a closed domain, the minimization can be attained. So we may assume that
$$\tilde{F}(r_1,\rho_1)=\min\{\tilde{F}(r,\rho):(r,\rho) \in S_\epsilon\times S_\epsilon\}.$$

We claim that $(r_1,\rho_1)$ is an interior point of $S_\epsilon\times S_\epsilon$.

We assume that
$$
\tilde{g}_1(r)=a\tilde{B}r^m+B_1e^{-\frac{2r\sin\frac{\pi}{2k}}{\epsilon}}
\,\,\text{and}\,\,
\tilde{h}_1(\rho)=b\tilde{C}\rho^m+B_2e^{-\frac{2\rho\sin\frac{\pi}{2k}}{\epsilon}}.
$$
By direct computation, we see that $\tilde{g}_1(r)$ attain the local minimization at
$$\bar{r}=\frac{m+o_\epsilon(1)}{2\sin\frac{\pi}{2k}}\epsilon\ln\frac{1}{\epsilon}.$$
We have that
$$\tilde{g}_1(\bar{r})=\Big(a\tilde{B}\Big(\frac{m}{2\sin\frac{\pi}{2k}}\Big)^m+o_\epsilon(1)\Big)\Big(\epsilon\ln\ds\frac{1}{\epsilon}\Big)^m,$$
$$\tilde{g}_1\Big(\frac{m-\tilde{\delta}}{2\sin\frac{\pi}{2k}}\epsilon\ln\frac{1}{\epsilon}\Big)=C\epsilon^{m-\tilde{\delta}},$$
and
$$\tilde{g}_1\Big(\frac{m+\tilde{\delta}}{2\sin\frac{\pi}{2k}}\epsilon\ln\frac{1}{\epsilon}\Big)
=\Big(a\tilde{B}\Big(\frac{m+\tilde{\delta}}{2\sin\frac{\pi}{2k}}\Big)^m+o_\epsilon(1)\Big)\Big(\epsilon\ln\ds\frac{1}{\epsilon}\Big)^m.$$

Similarly, $\tilde{h}_1(\rho)$ also attains the local minimization at
$$\bar{\rho}=\frac{m+o_\epsilon(1)}{2\sin\frac{\pi}{2k}}\epsilon\ln\frac{1}{\epsilon}.$$
And we also have
$$\tilde{h}_1(\bar{\rho})=\Big(b\tilde{C}\Big(\frac{m}{2\sin\frac{\pi}{2k}}\Big)^m+o_\epsilon(1)\Big)\Big(\epsilon\ln\ds\frac{1}{\epsilon}\Big)^m,$$
$$\tilde{h}_1\Big(\frac{m-\tilde{\delta}}{2\sin\frac{\pi}{2k}}\epsilon\ln\frac{1}{\epsilon}\Big)=C\epsilon^{m-\tilde{\delta}},$$
and
$$\tilde{h}_1\Big(\frac{m+\tilde{\delta}}{2\sin\frac{\pi}{2k}}\epsilon\ln\frac{1}{\epsilon}\Big)=\Big(b\tilde{C}
\Big(\frac{m+\tilde{\delta}}{2\sin\frac{\pi}{2k}}\Big)^m+o_\epsilon(1)\Big)\Big(\epsilon\ln\ds\frac{1}{\epsilon}\Big)^m.$$

And we may assume that
$$
\tilde{g}_2(r)=a\tilde{B}r^m+(B_1+o_\epsilon(1))e^{-\frac{2r\sin\frac{\pi}{2k}}{\epsilon}}\,\,
\text{and}\,\,
\tilde{h}_2(\rho)=b\tilde{C}\rho^m+B_2e^{-\frac{2\rho\sin\frac{\pi}{2k}}{\epsilon}}.
$$
By direct computation, 
we see that $\tilde{g}_2(r)$ attains the local minimization at
$$
\bar{r}=\frac{m+o_\epsilon(1)}{2\sin\frac{\pi}{2k}}\epsilon\ln\frac{1}{\epsilon}.
$$
We have that
$$
\tilde{g}_2(\bar{r})=\Big(a\tilde{B}\Big(\frac{m}{2\sin\frac{\pi}{2k}}\Big)^m+o_\epsilon(1)\Big)
\Big(\epsilon\ln\ds\frac{1}{\epsilon}\Big)^m,
$$
$$
\tilde{g}_2\Big(\frac{m-\tilde{\delta}}{2\sin\frac{\pi}{2k}}\epsilon\ln\frac{1}{\epsilon}\Big)
=C\epsilon^{m-\tilde{\delta}}
$$
and
$$
\tilde{g}_2\Big(\frac{m+\tilde{\delta}}{2\sin\frac{\pi}{2k}}\epsilon\ln\frac{1}{\epsilon}\Big)
=\Big(a\tilde{B}\Big(\frac{m+\tilde{\delta}}{2\sin\frac{\pi}{2k}}\Big)^m+o_\epsilon(1)\Big)
\Big(\epsilon\ln\ds\frac{1}{\epsilon}\Big)^m.
$$

Similarly, $\tilde{h}_2(\rho)$ also attains the local minimization at
$$
\bar{\rho}=\frac{m+o_\epsilon(1)}{2\sin\frac{\pi}{2k}}\epsilon\ln\frac{1}{\epsilon}.
$$
And we also have
$$
\tilde{h}_2(\bar{\rho})=\Big(b\tilde{C}\Big(\frac{m}{2\sin\frac{\pi}{2k}}\Big)^m+o_\epsilon(1)\Big)
\Big(\epsilon\ln\ds\frac{1}{\epsilon}\Big)^m,
$$
$$
\tilde{h}_2\Big(\frac{m-\tilde{\delta}}{2\sin\frac{\pi}{2k}}\epsilon\ln\frac{1}{\epsilon}\Big)
=C\epsilon^{m-\tilde{\delta}}
$$
and
$$\tilde{h}_2\Big(\ds\frac{m+\tilde{\delta}}{2\sin\frac{\pi}{2k}}\epsilon\ln\frac{1}{\epsilon}\Big)
=\Big(b\tilde{C}\Big(\ds\frac{m+\tilde{\delta}}{2\sin\frac{\pi}{2k}}\Big)^m+o_\epsilon(1)\Big)\Big(\epsilon\ln\ds\frac{1}{\epsilon}\Big)^m.$$

If $o_\epsilon(1)>0$, then
$$
\begin{array}{rl}
\tilde{F}(r_1,\rho_1)&\leq 2k\epsilon^3 \Big[\tilde{A} +\min\limits_{(r,\rho)\in S_\epsilon\times S_\epsilon}\{\tilde{g}_2(r)+\tilde{h}_2(\rho)+O(r^{m-1}\epsilon +\rho^{n-1}\epsilon)\}\Big]\\[5mm]
&\leq  2k\epsilon^3 \Big[\tilde{A} + \Big(a\tilde{B}\Big(\ds\frac{m}{2\sin\frac{\pi}{2k}}\Big)^m+b\tilde{C}\Big(\frac{m}{2\sin\frac{\pi}{2k}}\Big)^m
+o_\epsilon(1)\Big)\Big(\epsilon\ln\ds\frac{1}{\epsilon}\Big)^m\Big].\\[5mm]
\end{array}
$$
If $o_\epsilon(1)\leq 0$, then
$$
\begin{array}{rl}
\tilde{F}(r_1,\rho_1)&\leq  2k\epsilon^3 \Big[\tilde{A} + \min\limits_{(r,\rho)\in S_\epsilon\times S_\epsilon}\{\tilde{g}_1(r)+\tilde{h}_1(\rho)+O(r^{m-1}\epsilon +\rho^{n-1}\epsilon)\}\Big]\\[5mm]
&\leq 2k\epsilon^3 \Big[\tilde{A} +\Big((a\tilde{B}\Big(\ds\frac{m}{2\sin\frac{\pi}{2k}}\Big)^m+b \tilde{C} \Big(\ds\frac{m}{2\sin\frac{\pi}{2k}}\Big)^m+o_\epsilon(1)\Big)\Big(\epsilon\ln\ds\frac{1}{\epsilon}\Big)^m\Big].\\[5mm]
\end{array}
$$

Thus we get
\begin{equation}\label{3.8}
\tilde{F}(r_1,\rho_1)\leq 2k\epsilon^3 \Big[\tilde{A} +\Big((a\tilde{B}\Big(\ds\frac{m}{2\sin\frac{\pi}{2k}}\Big)^m+b \tilde{C} \Big(\ds\frac{m}{2\sin\frac{\pi}{2k}}\Big)^m+o_\epsilon(1)\Big)\Big(\epsilon\ln\ds\frac{1}{\epsilon}\Big)^m\Big].
\end{equation}

For convenience, we denote
 $r_{_l}:=\frac{m-\tilde{\delta}}{2\sin\frac{\pi}{2k}}\epsilon\ln\frac{1}{\epsilon},
r_{_r}:=\frac{m+\tilde{\delta}}{2\sin\frac{\pi}{2k}}\epsilon\ln\frac{1}{\epsilon}, \rho_{_l}:=\frac{m-\tilde{\delta}}{2\sin\frac{\pi}{2k}}\epsilon\ln\frac{1}{\epsilon},
\rho_{_r}:=\frac{m+\tilde{\delta}}{2\sin\frac{\pi}{2k}}\epsilon\ln\frac{1}{\epsilon}.$

If $o_\epsilon(1)> 0$, then
$$
\begin{array}{rl}
\tilde{F}(r_{_l},\rho)&\geq  2k\epsilon^3 \Big[\tilde{A} +\tilde{g}_1(r_{_l})+\tilde{h}_1(\rho)+O(r_{_l}^{m-1}\epsilon +\rho^{n-1}\epsilon)\}\Big]\\[5mm]&\geq 2k\epsilon^3 \Big[\tilde{A} + C\epsilon^{m-\tilde{\delta}}+\Big(b\tilde{C}\Big(\ds\frac{m}{2\sin\frac{\pi}{2k}}\Big)^m+o_\epsilon(1)\Big)
\Big(\epsilon\ln\ds\frac{1}{\epsilon}\Big)^m\Big].
\end{array}
$$
If $o_\epsilon(1)\leq 0$, then
$$
\begin{array}{rl}\tilde{F}(r_{_l},\rho)&= 2k\epsilon^3 \Big[\tilde{A} +\tilde{g}_2(r_{_l})+\tilde{h}_2(\rho)+O(r_{_l}^{m-1}\epsilon +\rho^{n-1}\epsilon)\}\Big]\\[5mm]
&\geq  2k\epsilon^3 \Big[\tilde{A} + C\epsilon^{m-\tilde{\delta}}+\Big(b\tilde{C}\Big(\ds\frac{m}{2\sin\frac{\pi}{2k}}\Big)^m
+o_\epsilon(1)\Big)\Big(\epsilon\ln\ds\frac{1}{\epsilon}\Big)^m\Big].
\end{array}
$$
 Therefore, we have
\begin{equation}\label{3.9}
\tilde{F}(r_{_l},\rho)\geq  2k\epsilon^3 \Big[\tilde{A} + C\epsilon^{m-\tilde{\delta}}+\Big(b\tilde{C}\Big(\ds\frac{m}{2\sin\frac{\pi}{2k}}\Big)^m+o_\epsilon(1)\Big)\Big(\epsilon\ln\ds\frac{1}{\epsilon}\Big)^m\Big].
\end{equation}

Similarly, we also have
\begin{equation}\label{3.10}
\tilde{F}(r_{_r},\rho)\geq 2k\epsilon^3 \Big[\tilde{A} +\Big(a\tilde{B}(\frac{m+\tilde{\delta}}{2\sin\frac{\pi}{2k}})^m+
b\tilde{C}\Big(\ds\frac{m}{2\sin\frac{\pi}{2k}}\Big)^m
+o_\epsilon(1)\Big)\Big(\epsilon\ln\ds\frac{1}{\epsilon}\Big)^m\Big],
\end{equation}
\begin{equation}\label{3.11}
\tilde{F}(r,\rho_{_l})\geq 2k\epsilon^3 \Big[\tilde{A} +C\epsilon^{m-\tilde{\delta}}+\Big(a\tilde{B}\Big(\ds\frac{m}{2\sin\frac{\pi}{2k}}\Big)^m+o_\epsilon(1)\Big)
\Big(\epsilon\ln\ds\frac{1}{\epsilon}\Big)^m\Big]
\end{equation}
and
\begin{equation}\label{3.12}
\tilde{F}(r,\rho_{_r})\geq 2k\epsilon^3 \Big[\tilde{A} +\Big(a\tilde{B}\Big(\ds\frac{m}{2\sin\frac{\pi}{2k}}\Big)^m
+b\tilde{C}(\frac{m+\tilde{\delta}}{2\sin\frac{\pi}{2k}})^m+o_\epsilon(1)\Big)\Big(\epsilon\ln\ds\frac{1}{\epsilon}\Big)^m\Big].
\end{equation}

From \eqref{3.8} to \eqref{3.12}, we can see that when $\epsilon$ is sufficiently small, the local  minimization of $\tilde{F}(r,\rho)$
can't be obtained at the boundary of $S_\epsilon \times S_\epsilon$. That is, $(r_1, \rho_1)$ is an interior point of $S_\epsilon \times S_\epsilon$.
Thus $(r_1, \rho_1)$ is a critical point of $\tilde{F}(r,\rho)$.
So $(\tilde{U}_{r_1}+\tilde{\varphi}(r_1,\rho_1),\tilde{V}_{\rho_1}+\tilde{\psi}(r_1,\rho_1))$ is a solution of \eqref{1.1}.
This completes the proof.

\appendix

\section{Energy estimate}
\def\theequation{A.\arabic{equation}}\makeatother
\setcounter{equation}{0}

In this section, we will give out some energy estimates of the approximate solutions.
Recall that
$$ x^j:=\Big(r\cos\frac{(j-1)\pi}{k},~r\sin\frac{(j-1)\pi}{k},~x_3\Big),~j=1,2,\cdots,2k,$$
$$y^j:=\Big(\rho\cos\frac{(2j-1)\pi}{2k},~\rho\sin\frac{(2j-1)\pi}{2k},~x_3\Big),j=1,2,\cdots,2k,$$
$$U_r(x)=\sum\limits_{j=1}^{2k}(-1)^{j-1}U_{x^j,\epsilon},
V_r(x)=\sum\limits_{j=1}^{2k}(-1)^{j-1}V_{x^j,\epsilon},
$$
$$
\tilde{U}_r=\sum\limits_{j=1}^{2k}(-1)^{j-1}U_{1,x^j,\epsilon},\,\,
\tilde{V}_\rho=\sum\limits_{j=1}^{2k}(-1)^{j-1}U_{2,y^j,\epsilon}
$$
and
$$
\begin{array}{rl}
I_\epsilon(u,v)&=\ds\frac{1}{2}\ds\int_{\R^3}\Big(\epsilon^2|\nabla u|^2 +P(x)u^2 +
\epsilon^2|\nabla v|^2 +Q(x)v^2\Big) -\frac{1}{4}\ds\int_{\R^3}\Big(\mu_1|u|^4+\mu_2|v|^4\Big)\\[5mm] &
\quad-\ds\frac{\beta}{2}\ds\int_{\R^3}u^2v^2.
\end{array}
$$

\begin{proposition}\label{proA1}
Assume that $(P)$ and $(Q)$ hold. Then we get the following energy estimate:
$$
I_\epsilon(U_{x^j,\epsilon},V_{x^j,\epsilon})=\epsilon^3\big[
A+aBr^m+bC_0r^n+O\big(r^{m-1}\epsilon+r^{n-1}\epsilon+e^{-\frac{(2-\tau)(1-\tau)r}{\epsilon}}\big)\big],
$$
where $a,b$ is given in  $(P)$ and $(Q)$, $\tau$ is determined in Lemma \ref{lemma2.5},  $A=\frac{1}{4}(\mu_1 \alpha^4+\mu_2\gamma^4+2\beta\alpha^2\gamma^2)\ds\int_{\R^3}w^4~dx$,
 $B=\frac{1}{2}\alpha^2\ds\int_{\R^3}w^2~dx$,  and $C_0=\frac{1}{2}\gamma^2\ds\int_{\R^3}w^2~dx.$
\end{proposition}
\begin{proof}
By direct computation, we have
\begin{equation}\label{A.1}
\begin{array}{rl}
I_\epsilon(U_{x^j,\epsilon},V_{x^j,\epsilon})
&=\ds\frac{1}{2}\ds\int_{\R^3}\big(\epsilon^2|\nabla U_{x^j,\epsilon}|^2 +U_{x^j,\epsilon}^2 +
\epsilon^2|\nabla V_{x^j,\epsilon}|^2 +V_{x^j,\epsilon}^2\big) \\[5mm]
 &\quad-\ds\frac{1}{4}\ds\int_{\R^3}\big(\mu_1|U_{x^j,\epsilon}|^4+\mu_2|V_{x^j,\epsilon}|^4\big)
-\frac{\beta}{2}\ds\int_{\R^3}U_{x^j,\epsilon}^2V_{x^j,\epsilon}^2\\[5mm]
&\quad+\ds\frac{1}{2}\ds\int_{\R^3}\big[(P(x)-1)U_{x^j,\epsilon}^2 +
(Q(x)-1)V_{x^j,\epsilon}^2\big]\\[5mm]
&=\ds\frac{1}{4}\ds\int_{\R^3}\big(\mu_1|U_{x^j,\epsilon}|^4+\mu_2|V_{x^j,\epsilon}|^4\big)+
\frac{\beta}{2}\ds\int_{\R^3}U_{x^j,\epsilon}^2V_{x^j,\epsilon}^2\\[5mm]
&\quad+\ds\frac{1}{2}\ds\int_{\R^3}\big[(P(x)-1)U_{x^j,\epsilon}^2 +
(Q(x)-1)V_{x^j,\epsilon}^2\big].\\[5mm]
\end{array}
\end{equation}

But
\begin{equation}\label{A.2}
\begin{array}{rl}
\ds\frac{1}{4}\ds\int_{\R^3}\big(\mu_1|U_{x^j,\epsilon}|^4+\mu_2|V_{x^j,\epsilon}|^4\big)
&=\ds\frac{\epsilon^3}{4}\ds\int_{\R^3}\big(\mu_1U^4+\mu_2V^4\big)\\[5mm]
&=\ds\frac{\epsilon^3}{4}(\mu_1\alpha^4+\mu_2\gamma^4)\ds\int_{\R^3}w^4
\end{array}
\end{equation}
and
\begin{equation}\label{A.3}
\begin{array}{rl}
\frac{\beta}{2}\ds\int_{\R^3}U_{x^j,\epsilon}^2V_{x^j,\epsilon}^2
=\ds\frac{\beta}{2}\epsilon^3\ds\int_{\R^3}U^2V^2
=\ds\frac{\beta}{2}\epsilon^3\alpha^2\gamma^2\ds\int_{\R^3}w^4.
\end{array}
\end{equation}

For any $m>1$ and any $0<d<1,$ we have
$$
|\epsilon y+x^j|^m=|x^j|^m\Big(1+O\Big(\frac{|\epsilon y|}{|x^j|}\Big)\Big), y \in B_{\frac{dr}{\epsilon}}(0).
$$
Since 
$$
P(r)=1+ar^m+O(r^{m+\theta})\quad \hbox{as}~r \to 0^+,
$$
we get
\begin{equation}\label{A.4}\begin{split}
&\frac{1}{2}\ds\int_{\R^3}(P(x)-1)U_{x^j,\epsilon}^2\\
&
\quad\quad=\frac{1}{2}\epsilon^3\ds\int_{\R^3}(P(\epsilon y+x^j)-1)U^2\\
&\quad\quad=\frac{1}{2}\epsilon^3\Big[\ds\int_{B_{\frac{(1-\tau)r}{\epsilon}}(0)}(P(\epsilon y+x^j)-1)U^2+
\ds\int_{B^c_{\frac{(1-\tau)r}{\epsilon}}(0)}(P(\epsilon y+x^j)-1)U^2\Big]\\
&\quad\quad=\frac{1}{2}\epsilon^3\Big[\ds\int_{B_{\frac{(1-\tau)r}{\epsilon}}(0)}(a|\epsilon y+x^j|^m+O(|\epsilon y+x^j|^{m+\theta}))U^2
+O(e^{-\frac{(2-\tau)(1-\tau)r}{\epsilon}})\Big]\\
&\quad\quad=\frac{1}{2}\epsilon^3\Big[\ds\int_{B_{\frac{(1-\tau)r}{\epsilon}}(0)}\bigg(a|x^j|^m\Big(1+O\Big(\frac{|\epsilon y|}{|x^j|}\Big)\Big)
+O(|x^j|^{m+\theta}\Big(1+O\Big(\frac{|\epsilon y|}{|x^j|}\Big)\Big))\bigg)U^2\\
&~~~\quad\quad~\quad\quad+O(e^{-\frac{(2-\tau)(1-\tau)r}{\epsilon}})\Big]\\
&\quad\quad=\frac{1}{2}\epsilon^3\Big[\ds\int_{B_{\frac{(1-\tau)r}{\epsilon}}(0)}ar^mU^2
+O\Big(\ds\int_{B_{\frac{(1-\tau)r}{\epsilon}}(0)}
r^{m-1}\epsilon|y|U^2\Big)+O(e^{-\frac{(2-\tau)(1-\tau)r}{\epsilon}})\Big]\\
&\quad\quad=\frac{1}{2}\epsilon^3\Big[\ds\int_{\R^3}ar^mU^2
-\ds\int_{B^c_{\frac{(1-\tau)r}{\epsilon}}(0)}ar^mU^2+O(r^{m-1}\epsilon)
+O(e^{-\frac{(2-\tau)(1-\tau)r}{\epsilon}})\Big]\\
&\quad\quad=\epsilon^3\Big[aBr^m+O(r^{m-1}\epsilon)
+O(e^{-\frac{(2-\tau)(1-\tau)r}{\epsilon}})\Big],\\
\end{split}
\end{equation}
where $\tau$ is a small positive constant.
Noting that 
$$
Q(r)=1+br^n+O(r^{n+\delta})\quad  \hbox{as}~r \to 0^+,
$$
by the same argument as above, we can get
\begin{equation}\label{A.5}
\frac{1}{2}\ds\int_{\R^3}(Q(x)-1)V_{x^j,\epsilon}^2)
=\epsilon^3\Big[bC_0r^n+O(r^{n-1}\epsilon)
+O(e^{-\frac{(2-\tau)(1-\tau)r}{\epsilon}})\Big].
\end{equation}

So combining \eqref{A.1}--\eqref{A.5}, we get $$I_\epsilon(U_{x^j,\epsilon},V_{x^j,\epsilon})=\epsilon^3\Big[
A+aBr^m+bC_0r^n+O\big(r^{m-1}\epsilon+r^{n-1}\epsilon+e^{-\frac{(2-\tau)(1-\tau)r}{\epsilon}}\big)\Big].$$

\end{proof}

\begin{proposition}\label{proA2}
Assume that $(P)$ and $(Q)$ hold. Then there exists a small constant $0<\sigma <\min\{\frac{1}{10},\min\{m,n\}-1\}$ and a positive constant $C$ such that
$$\begin{array}{rl}I_\epsilon(U_r,V_r)&=2k\epsilon^3\Big[A+aBr^m+bC_0r^n+C(\frac{\mu_1\alpha^4}{2}+\frac{\mu_2\gamma^4}{2}+\beta\alpha^2\gamma^2)e^{-\frac{2r\sin\frac{\pi}{2k}}{\epsilon}}\\
 &\quad\quad~\quad+ O\big(r^{m-1}\epsilon+r^{n-1}\epsilon+e^{-\frac{(1-\tau)(2-\tau)r}{\epsilon}}
 +e^{-\frac{(1+\sigma)2r\sin\frac{\pi}{2k}}{\epsilon}}\big)\Big],\\
\end{array}$$
where $\tau$ is defined in Lemma \ref{lemma2.5}.
\end{proposition}

\begin{proof}
 We know that
\begin{equation}\label{1.6}\begin{split}
I_\epsilon(U_r,V_r)
&=\sum\limits_{j=1}^{2k}I_\epsilon(U_{x^j,\epsilon},V_{x^j,\epsilon})\\
&\quad-\frac{\mu_1}{4}\ds\int_{\R^3}\Big[U_r^4-\sum\limits_{j=1}^{2k}U_{x^j,\epsilon}^4
-2\sum\limits_{i\neq j}(-1)^{i+j}U_{x^j,\epsilon}^3U_{x^i,\epsilon}\Big]\\
&\quad-\frac{\mu_2}{4}\ds\int_{\R^3}\Big[V_r^4-\sum\limits_{j=1}^{2k}V_{x^j,\epsilon}^4
-2\sum\limits_{i\neq j}(-1)^{i+j}V_{x^j,\epsilon}^3V_{x^i,\epsilon}\Big]\\
&\quad-\frac{\beta}{2}\ds\int_{\R^3}\Big[U_r^2V_r^2-\sum\limits_{j=1}^{2k}U_{x^j,\epsilon}^2V_{x^j,\epsilon}^2-\sum\limits_{i\neq j}(-1)^{i+j}V_{x^j,\epsilon}^2
U_{x^j,\epsilon}U_{x^i,\epsilon}-\sum\limits_{i\neq j}(-1)^{i+j}U_{x^j,\epsilon}^2
V_{x^,\epsilon}V_{x_i,\epsilon}\Big]\\
&\quad+\frac{1}{2}\sum\limits_{i\neq j}(-1)^{i+j}\ds\int_{\R^3}\Big[(P(x)-1)U_{x^j,\epsilon}U_{x^i,\epsilon}+(Q(x)-1)V_{x^j,\epsilon}V_{x^i,\epsilon}\Big].\\
\end{split}
\end{equation}

But there exists a small positive $0<\sigma <\min\{\frac{1}{10},\min\{m,n\}-1\}$ such that
\begin{equation}\label{A.7}\begin{split}
&-\frac{\mu_1}{4}\ds\int_{\R^3}\Big[U_r^4-\sum\limits_{j=1}^{2k}U_{x^j,\epsilon}^4
-2\sum\limits_{i\neq j}(-1)^{i+j}U_{x^j,\epsilon}^3U_{x^i,\epsilon}\Big]\\
&=\frac{\mu_1}{2}\ds\int_{\R^3}\Big[\sum\limits_{|i-j|=1 or 2k-1}U_{x^j,\epsilon}^3U_{x^i,\epsilon}+O\Big(\sum\limits_{1<|i-j|<2k-1}U_{x^j,\epsilon}^3U_{x^i,\epsilon}
+\sum\limits_{i\neq j}U_{x^j,\epsilon}^2U_{x^i,\epsilon}^2\Big)\Big]\\
&=\frac{\mu_1\alpha^4}{2}\ds\int_{\R^3}\sum\limits_{|i-j|=1 or 2k-1}w_{x^j,\epsilon}^3w_{x^i,\epsilon}+\epsilon^3 O\big(e^{-\frac{(1+\sigma)|x^1-x^2|}{\epsilon}}\big)\\
&=\epsilon^3 \Big(C\frac{\mu_1\alpha^4}{2}e^{-\frac{2r\sin\frac{\pi}{2k}}{\epsilon}}+O\big(e^{-\frac{(1+\sigma)|x^1-x^2|}{\epsilon}}\big)\Big)
\end{split}
\end{equation}

 Similarly, we have
\begin{equation}\label{1.8}
\begin{split}
&-\frac{\mu_2}{4}\ds\int_{\R^3}\Big[V_r^4-\sum\limits_{j=1}^{2k}V_{x^j,\epsilon}^4
-2\sum\limits_{i\neq j}(-1)^{i+j}V_{x^j,\epsilon}^3V_{x^i,\epsilon}\Big]\\
&=\frac{\mu_2\gamma^4}{2}\ds\int_{\R^3}\sum\limits_{|i-j|=1 or 2k-1}w_{x^j,\epsilon}^3w_{x^i,\epsilon}+\epsilon^3 O\Big(e^{-\frac{(1+\sigma)|x^1-x^2|}{\epsilon}}\Big),\\
&=\epsilon^3 \Big(C\frac{\mu_2\gamma^4}{2}e^{-\frac{2r\sin\frac{\pi}{2k}}{\epsilon}}+O\big(e^{-\frac{(1+\sigma)|x^1-x^2|}{\epsilon}}\big)\Big)
\end{split}
\end{equation}

and
\begin{equation}\label{1.9}
\begin{split}
&-\frac{\beta}{2}\ds\int_{\R^3}\Big[U_r^2V_r^2-\sum\limits_{j=1}^{2k}U_{x^j,\epsilon}^2V_{x^j,\epsilon}^2-\sum\limits_{i\neq j}(-1)^{i+j}V_{x^j,\epsilon}^2
U_{x^j,\epsilon}U_{x^i,\epsilon}-\sum\limits_{i\neq j}(-1)^{i+j}U_{x^j,\epsilon}^2
V_{x^j,\epsilon}V_{x^i,\epsilon}\Big]\\
&=\beta\alpha^2\gamma^2\ds\int_{\R^3}\sum\limits_{|i-j|=1 or 2k-1}w_{x^j,\epsilon}^3w_{x^i,\epsilon}+\epsilon^3
O\big(e^{-\frac{(1+\sigma)|x^1-x^2|}{\epsilon}}\big)\\
&=\epsilon^3 \big(C\beta\alpha^2\gamma^2e^{-\frac{2r\sin\frac{\pi}{2k}}{\epsilon}}
+O\big(e^{-\frac{(1+\sigma)|x^1-x^2|}{\epsilon}}\big)\big)
\end{split}
\end{equation}

 Combining \eqref{1.6}--\eqref{1.9} and Proposition \ref{proA1}, we can get

 \begin{equation}\begin{split}
 I_\epsilon(U_r,V_r)
 &=2k\epsilon^3\Big[A+aBr^m+bC_0r^n+C(\frac{\mu_1\alpha^4}{2}+\frac{\mu_2\gamma^4}{2}+\beta\alpha^2\gamma^2)e^{-\frac{2r\sin\frac{\pi}{2k}}{\epsilon}}\\
 &\quad\quad\quad\quad+ O\big(r^{m-1}\epsilon+r^{n-1}\epsilon+e^{-\frac{(1-\tau)(2-\tau)r}{\epsilon}}
 +e^{-\frac{(1+\sigma)2r\sin\frac{\pi}{2k}}{\epsilon}}\big)\Big].\\
 \end{split}
 \end{equation}

This completes the proof.
\end{proof}

\begin{lemma}\label{LemmaA3}
$$\ds\int_{\R^3}U_{1,x^i,\epsilon}^2U_{2,y^j,\epsilon}^2=\epsilon^3o_\epsilon(1)e^{-\frac{2|x^i-y^j|}{\epsilon}}.$$
\end{lemma}

\begin{proof}
Denote
$$
\Omega_1=\{x\in \R^3: |x-y^j|\geq|x-x^i|\},~\Omega_2=\{x\in \R^3: |x-y^j|\leq|x-x^i|\},
$$
$$
\omega_1=\{x\in \R^3: |x^i-y^j|\geq|x-y^j|\},~\omega_2=\{x\in \R^3: |x^i-y^j|\leq|x-y^j|\}
$$
and
$$
\omega^\prime_1=\Big\{x\in \omega_1: |x-y^j|\leq \epsilon \Big(\ln\frac{1}{\epsilon}\Big)^\frac{1}{3}\Big\},
~\omega^{\prime\prime}_1=\Big\{x\in \omega_1: |x-y^j|\geq \epsilon \Big(\ln\frac{1}{\epsilon}\Big)^\frac{1}{3}\Big\}.
$$
Then we have
$$
\ds\int_{\R^3}U_{1,x^i,\epsilon}^2U_{2,y^j,\epsilon}^2=\ds\int_{\Omega_1}U_{1,x^i,\epsilon}^2U_{2,y^j,\epsilon}^2+
\ds\int_{\Omega_2}U_{1,x^i,\epsilon}^2U_{2,y^j,\epsilon}^2.
$$

Since
we can estimate this term $\ds\int_{\Omega_1}U_{1,x^i,\epsilon}^2U_{2,y^j,\epsilon}^2$ similarly,
here we only estimate $\ds\int_{\Omega_2}U_{1,x^i,\epsilon}^2U_{2,y^j,\epsilon}^2.$

By the definition of $\Omega_2$, we can conclude that
$$|x-x^i|\geq \frac{1}{2}|x^i-y^j|, \quad ~\forall x \in \Omega_2.$$
Then we have
\begin{equation}\label{1.11}
\begin{split}
\ds\int_{\Omega_2\cap\omega_2}U_{1,x^i,\epsilon}^2U_{2,y^j,\epsilon}^2&\leq Ce^{-\frac{2|x^i-y^j|}{\epsilon}}
\ds\int_{\Omega_2\cap\omega_2}e^{-\frac{2|x-x^i|}{\epsilon}}\\
&\leq Ce^{-\frac{2|x^i-y^j|}{\epsilon}}
\ds\int_{\Omega_2\cap\omega_2}e^{-\frac{|x-x^i|}{\epsilon}}e^{-\frac{|x^i-y^j|}{2\epsilon}}.\\
&\leq C\epsilon^3e^{-\frac{5|x^i-y^j|}{2\epsilon}}\\
\end{split}
\end{equation}

$$
\ds\int_{\Omega_2\cap\omega_1}U_{1,x^i,\epsilon}^2U_{2,y^j,\epsilon}^2=\ds\int_{\Omega_2\cap\omega^\prime_1}U_{1,x^i,\epsilon}^2U_{2,y^j,\epsilon}^2
+\ds\int_{\Omega_2\cap\omega^{\prime\prime}_1}U_{1,x^i,\epsilon}^2U_{2,y^j,\epsilon}^2,
$$

\begin{equation}\label{A.8}
\begin{split}
\ds\int_{\Omega_2\cap\omega^\prime_1}U_{1,x^i,\epsilon}^2U_{2,y^j,\epsilon}^2&\leq C\frac{e^{-\frac{2|x^i-y^j|}{\epsilon}}}{|\frac{x^i-y^j}{\epsilon}|^2}
\ds\int_{\Omega_2\cap\omega^\prime_1}\frac{e^{-\frac{2(|x-x^i|+|x-y^j|-|x^i-y^j|)}{\epsilon}}}{|\frac{x-y^j}{\epsilon}|^2}\\
&\leq C\frac{e^{-\frac{2|x^i-y^j|}{\epsilon}}}{|\frac{x^i-y^j}{\epsilon}|^2}\epsilon^3
\ds\int_{|x|\leq (\ln\frac{1}{\epsilon})^\frac{1}{3}}\frac{e^{-2(|x|+|x-\frac{(x^i-y^j)}{\epsilon}|-|\frac{x^i-y^j}{\epsilon}|)}}{|x|^2}\\
&\leq C\frac{e^{-\frac{2|x^i-y^j|}{\epsilon}}}{|\frac{x^i-y^j}{\epsilon}|^2}\epsilon^3
\ds\int_{|x|\leq (\ln\frac{1}{\epsilon})^\frac{1}{3}}\frac{1}{|x|^2}\\
&\leq C\frac{e^{-\frac{2|x^i-y^j|}{\epsilon}}}{|\frac{x^i-y^j}{\epsilon}|^2}\epsilon^3\Big(\ln\frac{1}{\epsilon}\Big)^\frac{1}{3}\\
&\leq Ce^{-\frac{2|x^i-y^j|}{\epsilon}}\epsilon^3\frac{1}{\big(\ln\frac{1}{\epsilon}\big)^\frac{1}{3}}
\end{split}
\end{equation}
and
\begin{equation}\label{A.9}
\begin{split}
\ds\int_{\Omega_2\cap\omega^{\prime\prime}_1}U_{1,x^i,\epsilon}^2U_{2,y^j,\epsilon}^2&\leq Ce^{-\frac{2|x^i-y^j|}{\epsilon}}
\ds\int_{\Omega_2\cap\omega^{\prime\prime}_1}\frac{e^{-\frac{2(|x-x^i|+|x-y^j|-|x^i-y^j|)}{\epsilon}}}{|\frac{x-y^j}{\epsilon}|^4}\\
&\leq Ce^{-\frac{2|x^i-y^j|}{\epsilon}}\epsilon^3
\ds\int_{|x|\geq (\ln\frac{1}{\epsilon})^\frac{1}{3}}\frac{e^{-2(|x|+|x-\frac{(x^i-y^j)}{\epsilon}|-|\frac{x^i-y^j}{\epsilon}|)}}{|x|^4}\\
&\leq Ce^{-\frac{2|x^i-y^j|}{\epsilon}}\epsilon^3
\ds\int_{|x|\geq (\ln\frac{1}{\epsilon})^\frac{1}{3}}\frac{1}{|x|^4}\\
&\leq Ce^{-\frac{2|x^i-y^j|}{\epsilon}}\epsilon^3\frac{1}{\big(\ln\frac{1}{\epsilon}\big)^\frac{1}{3}}.
\end{split}
\end{equation}

From \eqref{A.8} and \eqref{A.9}, we can easily get
\begin{equation}\label{1.13}\begin{split}
\ds\int_{\Omega_2\cap\omega_1}U_{1,x^i,\epsilon}^2U_{2,y^j,\epsilon}^2=o_\epsilon(1)e^{-\frac{2|x^i-y^j|}{\epsilon}}\epsilon^3.
\end{split}
\end{equation}

Combining \eqref{1.11} and \eqref{1.13}, we can get
$$\ds\int_{\Omega_2}U_{1,x^i,\epsilon}^2U_{2,y^j,\epsilon}^2=o_\epsilon(1)e^{-\frac{2|x^i-y^j|}{\epsilon}}\epsilon^3.$$
With the same method, we can also obtain that
$$\ds\int_{\Omega_1}U_{1,x^i,\epsilon}^2U_{2,y^j,\epsilon}^2=o_\epsilon(1)e^{-\frac{2|x^i-y^j|}{\epsilon}}\epsilon^3.$$
So$$\ds\int_{\R^3}U_{1,x^i,\epsilon}^2U_{2,y^j,\epsilon}^2=o_\epsilon(1)e^{-\frac{2|x^i-y^j|}{\epsilon}}\epsilon^3.$$
This completes the proof.
\end{proof}

Using Lemma \ref{LemmaA3},   similar to Proposition \ref{proA1}, we can get the following Proposition.

\begin{proposition}\label{proA4}
Assume that $(P)$ and $(Q)$ hold. Then we get the following energy estimate:
$$\begin{array}{rl}
I_\epsilon(U_{1,x^j,\epsilon},U_{2,y^j,\epsilon})&=\epsilon^3\Big[\tilde{A}+
a\tilde{B}r^m+b\tilde{C}\rho^n-o_\epsilon(1)e^{-\frac{2\sqrt{(\rho-r\cos\frac{\pi}{2k})^2+(r\sin\frac{\pi}{2k})^2}}{\epsilon}}\\[5mm]
&\quad\quad+O\Big(e^{-\frac{(1-\tilde{\tau}_1)(2-\tilde{\tau}_1)r}{\epsilon}}
+e^{-\frac{(1-\tilde{\tau}_1)(2-\tilde{\tau}_1)\rho}{\epsilon}}+\rho^{n-1}\epsilon+r^{m-1}\epsilon
\Big)\Big],\end{array}$$
where $a,b$ is given in $(P)$ and $(Q)$, $\tilde{\tau}_1$ has been determined in Lemma \ref{lm3.3} , $\tilde{A}=\frac{1}{4}\ds\int_{\R^3}(\mu_1U_1^4+\mu_2U_2^4)~dx$,
$\tilde{B}=\frac{1}{2}\ds\int_{\R^3}U_1^2~dx$,  and $\tilde{C}=\frac{1}{2}\ds\int_{\R^3}U_2^2~dx.$
\end{proposition}

\begin{proof}
We know that

$$\begin{array}{rl}
I_\epsilon(U_{1,x^j,\epsilon},U_{2,y^j,\epsilon})
&
=\ds\frac{1}{2}\ds\int_{\R^3}\big(\epsilon^2|\nabla U_{1,x^j,\epsilon}|^2 +U_{1,x^j,\epsilon}^2 +
\epsilon^2|\nabla U_{2,y^j,\epsilon}|^2 +U_{2,y^j,\epsilon}^2\big)\\[5mm]
&\quad-\ds\frac{1}{4}\ds\int_{\R^3}\big(\mu_1|U_{1,x^j,\epsilon}|^4+\mu_2|U_{2,y^j,\epsilon}|^4\big)
-\ds\frac{\beta}{2}\ds\int_{\R^3}U_{1,x^j,\epsilon}^2U_{2,y^j,\epsilon}^2\\[5mm]
&\quad+\ds\frac{1}{2}\ds\int_{\R^3}\Big[(P(x)-1)U_{1,x^j,\epsilon}^2+(Q(x)-1)U_{2,y^j,\epsilon}^2\Big].\\[5mm]
\end{array}$$

Since $U_i$ is the unique radial solutions of the following problem
$$-\Delta u+u =\mu_iu^3,\, \max\limits_{x\in\R^3}u=u(0),u>0,$$
we have
$$
\begin{array}{rl}
&\ds\frac{1}{2}\ds\int_{\R^3}\big(\epsilon^2|\nabla U_{1,x^j,\epsilon}|^2 +U_{1,x^j,\epsilon}^2 +
\epsilon^2|\nabla U_{2,y^j,\epsilon}|^2 +U_{2,y^j,\epsilon}^2\big)\\[5mm]
&-\ds\frac{1}{4}\ds\int_{\R^3}\big(\mu_1|U_{1,x^j,\epsilon}|^4+\mu_2|U_{2,y^j,\epsilon}|^4\big)\\[5mm]
&=
\ds\frac{1}{4}\ds\int_{\R^3}\big(\mu_1|U_{1,x^j,\epsilon}|^4+\mu_2|U_{2,y^j,\epsilon}|^4\big)\\[5mm]
&=\ds\frac{1}{4}\epsilon^3\ds\int_{\R^3}\big(\mu_1U_1^4+\mu_2U_2^4\big).\\[5mm]
\end{array}
$$
Similar to \eqref{A.4}, noting that 
$$
P(r)=1+ar^m+O(r^{m+\theta})~\quad \hbox{as}~r \to 0^+
$$ 
and
$$
Q(r)=1+br^n+O(r^{n+\delta})~ \quad \hbox{as}~r \to 0^+,
$$ 
we can get that
\begin{equation}
\frac{1}{2}\ds\int_{\R^3}(P(x)-1)U_{1,x^j,\epsilon}^2 ~dx
=\epsilon^3\Big[a\tilde{B}r^m+O(r^{m-1}\epsilon)
+O\big(e^{-\frac{(2-\tilde{\tau}_1)(1-\tilde{\tau}_1)r}{\epsilon}}\big)\Big]
\end{equation}
and
\begin{equation}
\frac{1}{2}\ds\int_{\R^3}(Q(x)-1)U_{2,y^j,\epsilon}^2~dx
=\epsilon^3\Big[b\tilde{C}\rho^n+O(\rho^{n-1}\epsilon)
+O\big(e^{-\frac{(2-\tilde{\tau_1})(1-\tilde{\tau_1})\rho}{\epsilon}}\big)\Big],
\end{equation}

where $\tilde{\tau}_1>0$ is a constant.

From Lemma \ref{LemmaA3}, we have that
$$ \frac{\beta}{2}\ds\int_{\R^3}U_{1,x^j,\epsilon}^2U_{2,y^j,\epsilon}^2~dx=\beta\epsilon^3o_\epsilon(1)e^{-\frac{2|x^1-y^1|}{\epsilon}}.$$
Therefore, we have
$$
\begin{array}{rl}
I_\epsilon(U_{1,x^j,\epsilon},U_{2,y^j,\epsilon})&=\epsilon^3\Big[\tilde{A}+
a\tilde{B}r^m+b\tilde{C}\rho^n-o_\epsilon(1)e^{-\frac{2\sqrt{(\rho-r\cos\frac{\pi}{2k})^2+(r\sin\frac{\pi}{2k})^2}}{\epsilon}}\\[5mm]
&\quad\quad\,\,+O\big(e^{-\frac{(1-\tilde{\tau}_1)(2-\tilde{\tau}_1)r}{\epsilon}}
+e^{-\frac{(1-\tilde{\tau}_1)(2-\tilde{\tau}_1)\rho}{\epsilon}}+\rho^{n-1}\epsilon+r^{m-1}\epsilon
\big)\Big].
\end{array}
$$
We complete the proof.

\end{proof}

\begin{proposition}\label{proA5}
Assume that $(P)$ and $(Q)$ hold. Then there exist positive constants $B_1$ and $B_2$ such that
$$
\begin{array}{rl}
&I_\epsilon(\tilde{U}_r,\tilde{V}_\rho)=2k\epsilon^3\Big[\tilde{A}+
a\tilde{B}r^m+b\tilde{C}\rho^n+B_1e^{-\frac{2r\sin\frac{\pi}{2k}}{\epsilon}}
+B_2e^{-\frac{2\rho\sin\frac{\pi}{2k}}{\epsilon}}\\[5mm]
&\quad\quad\quad\quad\quad\quad\quad\quad+o_\epsilon(1)
e^{-\frac{2\sqrt{(\rho-r\cos\frac{\pi}{2k})^2+(r\sin\frac{\pi}{2k})^2}}{\epsilon}}
+O\big(e^{-\frac{(1-\tilde{\tau}_1)(2-\tilde{\tau}_1)r}{\epsilon}}
+e^{-\frac{(1-\tilde{\tau}_1)(2-\tilde{\tau}_1)\rho}{\epsilon}}
\\[5mm]
&\quad\quad\quad\quad\quad\quad\quad\quad
+\rho^{n-1}\epsilon+r^{m-1}\epsilon
+e^{-\frac{(1+\sigma)2r\sin\frac{\pi}{2k}}{\epsilon}}+e^{-\frac{(1+\sigma)2\rho\sin\frac{\pi}{2k}}{\epsilon}}\big)\Big],
\end{array}
$$
where $\sigma$ has been determined in Proposition \ref{proA2}.

\end{proposition}

\begin{proof} We can obtain that
\begin{equation}\label{A.17}
\begin{array}{rl}
I_\epsilon(\tilde{U}_r,\tilde{V}_\rho)
&=\ds\sum\limits_{j=1}^{2k}I_\epsilon(U_{1,x^j,\epsilon},U_{2,y^j,\epsilon})\\[5mm]
&\quad-\ds\frac{\mu_1}{4}\ds\int_{\R^3}\Big(|\tilde{U}_r|^4-\sum\limits_{j=1}^{2k}U_{1,x^j,\epsilon}^4-2\sum\limits_{i \neq j}(-1)^{i+j}U_{1,x^i,\epsilon}^3U_{1,x^j,\epsilon}\Big)\\[5mm]
&\quad-\ds\frac{\mu_2}{4}\ds\int_{\R^3}\Big(|\tilde{V}_\rho|^4-\sum\limits_{j=1}^{2k}U_{2,y^j,\epsilon}^4-2\sum\limits_{i \neq j}(-1)^{i+j}U_{2,y^i,\epsilon}^3 U_{2,y^j,\epsilon}\Big)\\[5mm]
&\quad-\ds\frac{\beta}{2}\ds\int_{\R^3}\Big(|\tilde{U}_r|^2|\tilde{V}_\rho|^2
-\sum\limits_{j=1}^{2k}U_{1,x^j,\epsilon}^2U_{2,y^j,\epsilon}^2\Big)\\[5mm]
&\quad+\ds\frac{1}{2}\sum\limits_{i \neq j}(-1)^{i+j}\ds\int_{\R^3}\Big[(P(x)-1)U_{1,x^i,\epsilon}U_{1,x^j,\epsilon}+(Q(x)-1)U_{2,y^i,\epsilon}U_{2,y^j,\epsilon}\Big].\\[5mm]
\end{array}
\end{equation}

Similar to \eqref{A.7}, we can get that there exist positive constants $B_1$ and $B_2$ such that
\begin{equation}
\begin{array}{rl}
&-\frac{\mu_1}{4}\ds\int_{\R^3}\Big(|\tilde{U}_r|^4-\sum\limits_{j=1}^{2k}U_{1,x^j,\epsilon}^4-2\sum\limits_{i \neq j}(-1)^{i+j}U_{1,x^i,\epsilon}^3U_{1,x^j,\epsilon}\Big)
\\[5mm]&=B_1e^{-\frac{2r\sin\frac{\pi}{2k}}{\epsilon}}\epsilon^3
+O\big(e^{-\frac{(1+\sigma)2r\sin\frac{\pi}{2k}}{\epsilon}}\big)\epsilon^3
\end{array}
\end{equation}
and
\begin{equation}
\begin{array}{rl}
&-\ds\frac{\mu_2}{4}\ds\int_{\R^3}\Big(|\tilde{V}_\rho|^4-\sum\limits_{j=1}^{2k}U_{2,y^j,\epsilon}^4-2\sum\limits_{i \neq j}(-1)^{i+j}U_{2,y^i,\epsilon}^3 U_{2,y^j,\epsilon}\Big)
\\[5mm]&=B_2e^{-\frac{2\rho\sin\frac{\pi}{2k}}{\epsilon}}\epsilon^3
+O\big(e^{-\frac{(1+\sigma)2\rho\sin\frac{\pi}{2k}}{\epsilon}}\big)\epsilon^3.
\end{array}
\end{equation}

On the other hand, we have
\begin{equation}
\begin{array}{rl}
&\frac{1}{2}\ds\int_{\R^3}(P(x)-1)U_{1,x^i,\epsilon}U_{1,x^j,\epsilon}\\[5mm]
&=\frac{1}{2}\ds\int_{B_{4r}(0)}(P(x)-1)U_{1,x^i,\epsilon}U_{1,x^j,\epsilon}+\frac{1}{2}\ds\int_{B^c_{4r}(0)}(P(x)-1)U_{1,x^i,\epsilon}U_{1,x^j,\epsilon}\\[5mm]
&\leq Cr^m\ds\int_{\R^3}U_{1,x^i,\epsilon}U_{1,x^j,\epsilon}+C\frac{1}{2}\ds\int_{B^c_{4r}(0)}(U_{1,x^i,\epsilon}^2+U_{1,x^j,\epsilon}^2)\\[5mm]
&\leq C\epsilon^3\big(r^m e^{-\frac{2r\sin\frac{\pi}{2k}}{\epsilon}}+e^{-\frac{3(2-\tilde{\tau}_1)r}{\epsilon}}\big)\\[5mm]
&=O\epsilon^3\big(r^{2m}+e^{-\frac{4r\sin\frac{\pi}{2k}}{\epsilon}}+e^{-\frac{3(2-\tilde{\tau}_1)r}{\epsilon}}\big).
\end{array}
\end{equation}
Similarly, we have
\begin{equation}
\frac{1}{2}\ds\int_{\R^3}(Q(x)-1)U_{2,y^i,\epsilon}U_{2,y^j,\epsilon}
=O\epsilon^3\big(\rho^{2n}+e^{-\frac{4\rho\sin\frac{\pi}{2k}}{\epsilon}}+e^{-\frac{3(2-\tilde{\tau}_1)\rho}{\epsilon}}\big).
\end{equation}
Then
\begin{equation}\label{A.22}
\begin{array}{rl}
\Big|\ds\frac{\beta}{2}\ds\int_{\R^3}\Big(|\tilde{U}_r|^2|\tilde{V}_\rho|^2
-\sum\limits_{j=1}^{2k}U_{1,x^j,\epsilon}^2U_{2,y^j,\epsilon}^2\Big)\Big|
&\leq C\sum\limits_{j=1}^{2k}\ds\int_{\R^3}U_{1,x^j,\epsilon}^2U_{2,y^j,\epsilon}^2\\[5mm]
&=\epsilon^3o_\epsilon(1)e^{-\frac{2\sqrt{(\rho-r\cos\frac{\pi}{2k})^2+(r\sin\frac{\pi}{2k})^2}}{\epsilon}}.
\end{array}
\end{equation}

From \eqref{A.17}--\eqref{A.22} and Proposition \ref{proA4}, we can easily have that
$$
\begin{array}{rl}
&I_\epsilon(\tilde{U}_r,\tilde{V}_\rho)=2k\epsilon^3\Big[\tilde{A}+
a\tilde{B}r^m+b\tilde{C}\rho^n+B_1e^{-\frac{2r\sin\frac{\pi}{2k}}{\epsilon}}
+B_2e^{-\frac{2\rho\sin\frac{\pi}{2k}}{\epsilon}}\\[5mm]
&\quad\quad\quad\quad\quad\quad\quad\quad+o_\epsilon(1)e^{-\frac{2\sqrt{(\rho-r\cos\frac{\pi}{2k})^2
+(r\sin\frac{\pi}{2k})^2}}{\epsilon}}+O\big(e^{-\frac{(1-\tilde{\tau}_1)(2-\tilde{\tau}_1)r}{\epsilon}}
+e^{-\frac{(1-\tilde{\tau}_1)(2-\tilde{\tau}_1)\rho}{\epsilon}}\\[5mm]
&\quad\quad\quad\quad\quad\quad\quad\quad+\rho^{n-1}\epsilon+r^{m-1}\epsilon
+e^{-\frac{(1+\sigma)2r\sin\frac{\pi}{2k}}{\epsilon}}+e^{-\frac{(1+\sigma)2\rho\sin\frac{\pi}{2k}}{\epsilon}}\big)\Big].
\end{array}
$$
This completes the proof.
\end{proof}

\end{document}